\documentclass[a4paper,11pt]{article}
\usepackage{jheppub} 
\usepackage[only,mapsfrom]{stmaryrd}
\usepackage{amsthm,tikz,tikz-cd}
\usetikzlibrary{decorations.pathreplacing}
\input{macros.sty}

\arxivnumber{2312.15958} 

\title{Category of SET orders}

\author{Tian Lan,}
\author{Gen Yue,}
\author{and Longye Wang}
\affiliation{Department of Physics, The Chinese University of Hong Kong,\\ Shatin, New Territories, Hong Kong, China}
\emailAdd{tlan@cuhk.edu.hk}

\abstract{We propose the representation principle to study physical systems with a given symmetry. In the context of symmetry enriched topological orders, we give the appropriate representation category, the category of SET orders, \GY{ which include SPT orders and symmetry breaking orders as special cases}. For fusion $n$-category symmetries, we show that the category of SET orders encodes almost all information about the interplay
between symmetry and topological orders, in a natural and canonical way. These information include defects and boundaries of SET orders, symmetry charges, explicit and spontaneous symmetry breaking, stacking of SET orders, gauging of
generalized symmetry, as well as quantum currents (SymTFT or symmetry TO). We also provide a detailed categorical algorithm to compute the generalized gauging. In particular, we proved that gauging is always reversible, as a special type of Morita-equivalence. The explicit data for ungauging, the inverse to gauging, is given.}

\begin{document}
\maketitle
\flushbottom
\section{Introduction}

We would like to propose a simple principle to study physical systems with a given symmetry

\begin{center}
\framebox{
    \parbox{.8\columnwidth}{\textbf{Representation Principle} ---\\
    Physical systems with a given symmetry form representation categories.}
}
\end{center}

This principle is widely accepted when studying the symmetry of mathematical structures, however, many physical applications seemed unsuccessful and obtained only partial results. Indeed, we need to carefully reinterpret ``symmetry'' as well as ``representation category''. The key point is to represent the symmetry in a category sufficiently complete that includes all relevant physical observables.  Simple models such as the category of Hilbert spaces, are not strong enough for such job. The unsuccessful applications of the principle are most likely due to the incompleteness of the representation category. Variants of higher categories are required, and the notion of symmetry also needs to be generalized correspondingly, to match the nature of the higher category of physical observables.

In this paper, for concreteness and mathematical rigor, we will restrict ourselves to ``topological ordered phase''~\cite{KW1405.5858,Wen89,Wen90,Kit0506438,Wen1506.05768,LKW1704.04221,LW1801.08530,Joh2003.06663} with ``finite generalized symmetry''~\cite{GKSW1412.5148,JW1905.13279,JW1912.13492,KLW+2005.14178,JW2106.02069,TW1912.02817,TW2106.12577,CW2203.03596,CW2205.06244,FMT2209.07471,ABE+2112.02092,BS2304.02660,BS2305.17159,Sha2308.00747} and elaborate on the appropriate representation category. We will show that, equipped with the power of modern mathematics, the representation category encodes almost all information about the interplay between symmetry and topological orders, i.e., symmetry enriched topological (SET) orders~\cite{CGLW1106.4772,BBCW1410.4540,MR1212.0835,CBVF1403.6491,HBFL1606.07816,CGJQ1606.08482,LKW1507.04673,LKW1602.05946,LKW1602.05936,BJLP1811.00434,KLW+2003.08898}. These information include defects and boundaries of SET orders (in particular symmetry charges), symmetry breaking, stacking of SET orders, gauging of generalized symmetry, as well as quantum currents~\cite{LZ2305.12917} (i.e., the background category in the enriched category description of quantum liquids~\cite{KZ1705.01087,CJK+1903.12334,KZ1905.04924,KZ1912.01760,LY2208.01572}, also called categorical symmetry~\cite{JW1905.13279,JW1912.13492,KLW+2005.14178,JW2106.02069}, SymTFT~\cite{ABE+2112.02092,BS2304.02660,BS2305.17159} or symmetry TO~\cite{CW2203.03596,CW2205.06244}).
We thus have a nice story in the context of SET orders, that everything can be derived from a single formula (see Definition~\ref{def.set}),
\begin{center}
\framebox{
\parbox{.3\columnwidth}{\centering $\Fun(\Sigma\cT,\cX).$}
}
\end{center}
\GY{In Newtonian mechanics, when a force $F$ acts on a point mass $m$, the point mass will obtain an acceleration $a=F/m$ according to Newtonian 2nd law. In analogue to this, when a symmetry $\cT$ acts on the underlying topological orders whose anomaly is characterized by $\cX$, all the SET orders (including SPT orders and symmetry breaking orders) form the category $\Fun(\Sigma\cT,\cX)$.} (See Table~\ref{tab.ana}.)

Although the principle is simple, the relevant mathematics is by no means simple. For the mechanics and electromagnetism case, we need the mathematics of (multi-variable) calculus and differential geometry to simplify the physics to a single formula. For the case of SET orders, the necessary mathematics is that about ``finite dimensional'' higher vector spaces, i.e., Karoubi complete (or condensation complete) higher linear categories which are dualizable~\cite{GJ1905.09566,Joh2003.06663,KZ2011.02859,KZ2107.03858}. Here Karoubi completion or condensation completion is the correct notion for the completion with respect to all topological observables. As these higher categories are the categorified version of vector spaces, we call the corresponding mathematicas ``higher linear algebra''. We give a concise introduction to higher linear algebra in Section~\ref{sec.hl}, taking a practical perspective and wishing to provide a toolbox for those interested in the applications more than the foundations of these mathematics.
\begin{center}
\begin{table}[ht]
\centering
    \begin{tabular}{|c|c| c|}
    \hline
         &Newtonian mechanics  &SET orders \\
       \hline
     Force/Symmetry & $F$  & $\cT$ \\
     \hline
     Matter & $m$  & $\cX$ \\
     \hline
    Acceleration/Category of SET Orders & $a=F/m$  &$\Fun(\Sigma \cT, \cX)$ \\
     \hline
     Required mathematics & Calculus & Higher linear algebra\\
     \hline
    \end{tabular}
    \caption{Analogy between Newtonian mechanics and SET orders.}
    \label{tab.ana}
\end{table}    
\end{center}

Let's begin with a finite global unitary symmetry on a quantum system with Hilbert space $V$. Such a symmetry requires us to pick global (codimension 1 in spacetime) unitary operators $\{U_g \in \End(V), g\in G\}$ whose multiplication rules are controlled by an abstract finite group $G$, 
\[ U_g U_h=U_{gh},\] 
such that the dynamics (Hamiltonian, Lagrangian, etc.) of the system is invariant with respect to $U_g$. At the (most rough) linear algebra level, we get a usual group representation $G\to\End(V),\ g\mapsto U_g$, equivalently a functor in $\Fun(BG,\ve)=:\Rep G$. However, this usual group representation neglects too much information of the system and fails to be a complete description.

For a gapped topological system, we should at least collect all the topological operators (defects, excitations, etc.). The topological operators are graded by their spacetime dimensions, and if we technically assume that the types of operators are finite (in appropriate sense), all the topological operators, in $n$+1D, of codimension $1,2,\dots,n+1$ form a (multi-)fusion $n$-category. We will denote this (multi)-fusion $n$-category of operators in a \emph{bare} (i.e., symmetry is not yet added) system, by $\cB$.

Now even if we consider the same finite group $G$ as the global symmetry, in different total spacetime dimensions, $G$ should be regarded as different symmetries. The reason is that the corresponding operators $U_g$ must be codimension 1 in spacetime, which depends on the total dimension sensitively. In order to match the spacetime dimension, we should promote $G$ to a monoidal $n$-category and consider a monoidal $n$-functor $G\to \cB$.

Next we want to upgrade $G\to \cB$ to a ``representation", in other words, we want a category $\cX$ and $\cB=\Hom_\cX(V,V)$ for some object $V\in \cX$, then $G\to \cB=\Hom_\cX(V,V)$ qualifies as a representation of $G$ on $V\in \cX$. Clearly, a minimal choice would be $\cX=\Sigma\cB=\mathrm{Kar}(B\cB)$. The corresponding representation category is then
\[ \Fun(BG, \Sigma\cB)\cong \Fun(\Sigma \ve[n]_G, \Sigma\cB).\]
Because $\Sigma\cB$ is Karoubi complete, we can replace $BG$ for the Karoubi completion of its linearized version $\Sigma\ve[n]_G$. The monoidal $n$-functor $G\rightarrow \cB$ corresponds to the representation $G\rightarrow \Hom_{\Sigma\cB}(\bullet,\bullet)=\cB$.

It is straightforward to generalize from $G$ to generalized (finite non-invertible) symmetry, a fusion $n$-category $\cT$, which abstractly encodes the multiplication rules of symmetry \emph{transformations}.
\begin{definition}\label{def.set}
  The $( n+1 )$-category of topological orders (with \emph{anomaly} $\cX\in \ve[n+2]$) enriched by symmetry (fusion $n$-category) $\cT$ is
  \[ \Fun(\Sigma\cT,\cX)=\Hom_{\ve[n+2]}(\Sigma\cT,\cX).\]
  For short, we call $\Fun(\Sigma\cT,\cX)$ the category of SET orders. It is the appropriate representation category in the context of SET order.
\end{definition}
\begin{remark}
When the symmetry $\cT$ (S in SET) and the (anomaly of) topological orders $\cX$ (T in SET) are specified, there is a unique representation category. For simplicity, we will not carry over the choice of $\cT$ and $\cX$ everywhere in our terminology. When we say ``\emph{the} category of SET orders'', we assume that both $\cT$ and $\cX$ are at least implicitly specified.
\end{remark}

The target category $\cX$ ought to be a category of physical systems. Based on the general discussion in Section~\ref{sec.hl}, we know for each object $x\in \cX$, $(\cX,x)$ is a potentially anomalous $n$+1D topological order whose defects form the (multi-)fusion $n$-category $\cB:=\Omega(\cX,x)$. For indecomposable $\cX$, the center, physically the bulk, of $\cB$ is $Z_1(\cB)=\Omega\Fun(\cX,\cX)$ and this is why we say $\cX$ characterizes the anomaly (bulk of $(\cX,x)$). The anomaly-free bare theories $\cB$ with $Z_1(\cB)=\ve[n]$ would be of particular interest, which is equivalent to picking \emph{invertible} objects $\cX\in\ve[n+2]$, i.e. $Z_0(\cX)=\Fun(\cX,\cX)=\cX\bt\cX^\op=\ve[n+1]$. When $\cX$ is invertible, it must be $\Sigma^2\cC$ for some non-degenerate braided fusion $( n-1 )$-category $\cC$, and $\cX=\Sigma^2\cC$ characterizes the Witt-class (a weaker version of anomaly) of $\cC$: for any simple $x\in \Sigma^2\cC$,  $\Omega^2(\Sigma^2\cC,x)$ is Witt-equivalent to $\cC$ and all Witt-equivalents ones are obtained this way. See Remark~\ref{rem.anomaly} for more details.

As an example, consider the 2+1D SET orders which have been extensively studied in the literature~\cite{KLW+2003.08898,BBCW1410.4540,BJLP1811.00434,CGJQ1606.08482,LKW1507.04673,LKW1602.05946,LKW1602.05936,HBFL1606.07816,CBVF1403.6491,MR1212.0835}. Since 2+1D topological orders are believed to be described by modular tensor categories (MTC), it is tempting to guess that SET orders should be described by a group $G$ represented in the 2-category $\mathbf{MTC}$ of modular tensor categories, braided functors and natural transformations, i.e., a representation corresponds to a monoidal functor $\rho: G\to \Aut^\text{br}(\cC)$ for some MTC $\cC$. Although this characterization already encodes rich information including the \emph{symmetry fractionalization}, it turns out to be another misuse of the representation principle due to incompleteness, since $\rho:G\to \Aut^\text{br}(\cC)$ can possess t'Hooft anomaly (or gauge anomaly)~\cite{BBCW1410.4540,CBVF1403.6491}. By~\cite{ENO0909.3140,BBCW1410.4540}, $\rho:G\to \Aut^\text{br}(\cC)$ can be extended to $\phi:\ve[2]_G^{O(\rho)}\to \Sigma\cC$ where $O(\rho)$ is an obstruction class depending on $\rho$ valued in $H^4(G,U(1))$. Only when $O(\rho)=1\in H^4(G,U(1))$, $\rho$ is free from 't Hooft anomaly. In this case $\Sigma\phi$ is an object in $\Fun(\Sigma\ve[2]_G,\Sigma^2\cC)$ which agrees with our general Definition~\ref{def.set}. We will also recover the case with 't Hooft anomaly, $\ve[2]_G^{\omega_4}\to \Sigma\cC$, as a defect between \GY{the $3$+1D trivial $G$-SPT, and the 3+1D $G$-SPT order labelled by $\omega_4$} (see Example~\ref{eg.3dset}). There are also alternative descriptions of 2+1D SET orders in terms of gauging, or minimal modular extensions~\cite{LKW1507.04673,LKW1602.05936,LKW1602.05946}. Our framework also covers these approaches and provides an explanation why the description in terms of gauged theories contains the same information as the original one.

A major part of this paper is devoted to the discussion of generalized gauging. 
Gauging is a somewhat mysterious procedure that turns a system with global symmetry into a gauge theory. It can be applied to both field theories and lattice systems. With the understanding that the gauged theory admits a generalized (non-invertible) symmetry exhibited by the fusion of gauge charges, one can gauge such generalized symmetry and go back to the original theory, which is called ungauging. However, the ungauging procedure is known only for some special cases (see~\cite{Vaf89,GKSW1412.5148,BT1704.02330} and references therein; see also~\cite{HBFL1606.07816} for a nice construction on lattice models). One of the main points of this paper is to set up a general framework for gauging so that it is always reversible.

Technically, to gauge a global symmetry, one first introduces a spacetime dependent gauge field that couples to the symmetry charges (promoting global symmetry to gauge symmetry) and then introduce proper interactions to make the gauge field fluctuate. Such interactions energetically prefer the states with less gauge charges and fluxes. From a macroscopic view, introducing gauge field is proliferating symmetry defects and the interactions prefer the symmetry defects to condense. Therefore, we can roughly say that gauging is the condensation of symmetry defects. This point of view is more manifest in quantum field theory, where the partition function of the gauged theory is exactly obtained by inserting symmetry defects in the original theory and summing over all symmetry defect configurations~\cite{GKSW1412.5148,RSS2204.02407}.

On the other hand, there is a parallel theory of condensation of anyons and more generally, higher dimensional topological defects, developed based on the algebra and representation theory in tensor category and higher category theories~\cite{Kon1307.8244,GJ1905.09566,KZ2011.02859,KZ2107.03858,KZZZ2403.07813}. 
From the point of view of category of SET orders $\Fun(\Sigma\cT,\cX)$, given $F\in \Fun(\Sigma\cT,\cX)$, its restriction to the distinguished object $\bullet$ gives the representation (symmetry assignment) $\Omega_\bullet F: \cT\to \Omega(\cX,F(\bullet))=:\cB$. Gauging then simply means restricting $F$ to a different object $A\in \Sigma\cT$, $\Omega_A F: \bmd A A \cT \to  \bmd{F(A)}{F(A)}\cB$, i.e., symmetry defects labeled by $A$ are condensed in the bare theory $\cB$ and one get a new bare theory $ \bmd{F(A)}{F(A)}\cB$ equipped with a new gauge symmetry $\bmd A A \cT$. Mathematically, $\bmd A A \cT$ denotes the $A$-$A$-bimodules in $\cT$, and similar for $\bmd{F(A)}{F(A)}\cB$. We will expand these details and clarify the categorical data needed to perform a gauging:
\begin{enumerate}
    \item the symmetry assignment of the original theory to be gauged and
    \item the symmetry defects to be condensed (the precise way to perform gauging).
\end{enumerate}
The first point tends to be overlooked, or implicitly assumed in microscopic constructions. We emphasize its role and show that different symmetry assignments indeed lead to different gauged theories. 
The second point is a natural generalization.
\GY{With the research on generalized symmetry in recent years, instead of just ordinary gauging or condensing all the symmetry defects, people have realized the importance of gauging generalized symmetris in different ways. For example,  non-invertible topological defects can be obtained from higher gauging of higher-form symmetry~\cite{RSS2204.02407}; SET orders and theories with generalized Kramers-Wannier dualities can be obtained by gauging a sub-group symmetry~\cite{Li2023PartialGauging,delcamp2024}; it was also realized that the Kennedy-Tasaki transformation can be viewed as a twisted gauging~\cite{LOZ2301.07899,SS2404.01369,LSY2405.14939}.
We can incorporate the condensation theory in higher categories and consider more exotic ways of gauging, including ordinary gauging, twisted gauging, and partial gauging (see Remark~\ref{rem.gtype} for a clarification on these terminologies). Note that some partial gauging can be obtained by assigning a smaller subsymmetry and completely gauging the subsymmetry, but in general partial gauging may not always arise from a subsymmetry (see Remark~\ref{rem.gsub} and Example~\ref{eg.pgnonsub}).
}
Moreover, we clarify the data for the inverse of a gauging procedure and unravel the fact that the original theory and the gauged theory are Morita equivalent. We also give a series of examples covering known variants of gauging.

\section{Higher linear algebra}\label{sec.hl}

In this section, we outline the necessary techniques used in this paper. Most mathematical results are due to~\cite{GJ1905.09566,KZ2011.02859,KZ2107.03858}; we reorganize them for a better and more concise exposition. As in~\cite{GJ1905.09566}, we use $n$-category to mean a weak $n$-category, and assume that the basic categorical formulations such as functors, (higher) natural transformations, Yoneda lemma, limits and colimits, etc., are ready for $n$-categories. We will pack all the information categorically into hom-categories or functor categories without explicit referring to the higher structures and coherence conditions. When a concrete model is present ($n=2,3$, i.e., bicategory and tricategory), we can work explicitly with examples.
Given an $n$-category $\cC$ and two $k$-morphisms $A,B$ in $\cC$, we use $\Hom_\cC(A,B)$ to denote the $(n-k-1)$-category of morphisms from $A$ to $B$. $\cC^{\op k}$ denotes the $n$-category obtained from $\cC$ by reversing all the $k$-morphisms, and in particular $\cC^\op$ means $\cC^{\op1}$.

Unless otherwise specified, all $n$-categories are assumed to be Karoubi complete $\C$-linear~\cite{GJ1905.09566,KZ2011.02859} $n$-categories, and functors are linear higher functors. 
We take the perspective that everything happens in the higher vector spaces, and we are simply doing higher linear algebra.

    Let $\mathrm{KarCat_n}$ denote the $(n+1)$-category of Karoubi-complete $\C$-linear $n$-categories. $\mathrm{KarCat_n}$ admits a symmetric monoidal structure, denoted by $\bt$, which can be characterized by the following universal property: For $\cC,\cD\in\mathrm{KarCat_n}$, $\cC\bt\cD$ is the object representing the functor $\Hom_{\mathrm{KarCat_n}}(\cC,\Hom_{\mathrm{KarCat_n}}(\cD,-)),$
    \begin{align*}
    \forall \cX\in\mathrm{KarCat_n}, \quad &\Hom_{\mathrm{KarCat_n}}(\cC\bt\cD,\cX)
    \cong \Hom_{\mathrm{KarCat_n}}(\cC,\Hom_{\mathrm{KarCat_n}}(\cD,\cX)).
    \end{align*}
\begin{definition}
    An $(n+1)$-vector space, or separable $n$-category~\cite{KZ2011.02859}, is a Karoubi complete $\C$-linear $n$-category which is dualizable with respect to $\bt$. The full subcategory of $\mathrm{KarCat_n}$ consisting of dualizable objects is denoted by $\ve[n+1]$.
     For two $(n+1)$-vector spaces $\cC,\cD$, we denote by $\Fun(\cC,\cD):=\Hom_{\ve[n+1]}(\cC,\cD)$ the $n$-category of linear functors.
\end{definition}
\begin{remark}
    A $1$-vector space is a usual finite dimensional vector space over $\C$. A $2$-vector space is a finite semi-simple $\C$-linear 1-category. 
\end{remark}
Many properties of $1$-vector spaces generalize to higher cases~\cite{KZ2011.02859}:
\begin{proposition}
Let $\cC,\cD$ be $(n+1)$-vector spaces:
    \begin{enumerate}
        \item $\Fun(\cC,\ve[n])$ is dual to $\cC$, i.e., the pairing $\Fun(\cC,\ve[n])\bt \cC\to \ve[n]$, $F\bt V\mapsto F(V)$ has a copairing $\ve[n]\to\cC\bt\Fun(\cC,\ve[n])$ satisfying zig-zag identities (and higher coherence conditions).
        \item $\Fun(\ve[n],\cC)\cong \cC$.
        \item $\cC^\op \cong \Fun(\cC,\ve[n]),$ where the equivalence is exactly the Yoneda embedding.
        \item $\Fun(\cC,\cD)\cong \cD\bt\cC^\op$.
        \item $\ve[n]$ is the unit with respect to $\bt$, $\ve[n]\bt\cC\cong \cC\bt \ve[n]\cong \cC$.
    \end{enumerate}
\end{proposition}
It is understood that $\ve[0]:=\C$. The above for $n=0$ are well known and useful results of finite dimensional vector spaces.

    We give one physical interpretation\footnote{It is not necessarily the only interpretation.} to the $(n+1)$-category $\ve[n+1]$ of $(n+1)$-vector spaces. Before going in to algebraic details, we can think the objects and (higher) morphisms just as labels or names of topological (gapped) defects of various dimensions. There are objects (0-morphisms), 1-morphisms, till $( n+1 )$-morphisms, so the spacetime dimension should also be $n+1$. A $k$-morphism labels a codimension $k$ defect:
    \begin{itemize}
        \item  An object $\cX$ in $\ve[n+1]$ labels an $n$+1D ``defect" (actually an $n$+1D phase).
        \item A 1-morphism $f\in\Fun(\cX,\cY)=\Hom_{\ve[n+1]}(\cX,\cY)$ labels an $n$D defect between $\cX$ and $\cY$.
        \item A 2-morphism $\alpha\in \Hom_{\ve[n+1]}(f,g)$ labels an $n{-}1$D defect between $f$ and $g$;
        \item So on and so forth, an $(n+1)$-morphism labels a $0$D defect in spacetime, which is physically a (equivalent class of) local operator.
    \end{itemize}
     In particular $\ve[n]$ labels the trivial phase, or the ``vacuum". The tensor product $\bt$ is physically stacking of defects. By inductive construction, $\Hom(A\bt B, X\bt Y)=\Hom(A,X)\bt\Hom(B,Y).$ Due to the property $\Fun(\ve[n],\cC)\cong \cC$, we also have physical interpretations to the object and morphisms \emph{inside} $\cC$, i.e.:
     \begin{itemize}
         \item An object $x\in \cC$ is identified with a 1-morphism $x\in \Fun(\ve[n],\cC)$, and thus a defect between $\ve[n]$ and $\cC$, in other words, a topological boundary condition of $\cC$.
         \item Higher morphisms are similarly higher codimensional defects on the boundary of $\cC$.
     \end{itemize}
     Therefore, we can say that an $(n+1)$-vector space $\cC$ is an $n$-category of boundary conditions of $\cC$.
     It is inspiring to rethink $\ve:=\ve[1]$, the usual linear algebra, or finite dimensional quantum mechanics, in such picture. A finite dimensional Hilbert space $V$ is an object in $\ve$, and it really labels a $0$+1D quantum system. The quantum states in $V$ are identified with the operators from $\C$ to $V$, and thus the $0$+0D boundary conditions of $V$. Informally speaking, we already have boundary-bulk correspondence (see also Remark~\ref{rem.center}) built-in in quantum mechanics.

The main feature by which higher vector spaces contrast with $1$-vector space is that the shift of categorical level via looping and delooping, which is also closely related to the hierarchy of monoidal structure.
\begin{definition}
    [$E_k$-multi-fusion $n$-category]
    An $E_0$-multi-fusion $n$-category $(\cC,\bullet)$ is an $(n+1)$-vector space $\cC$ with a distinguished object $\bullet\in \cC$.
    An $E_0$-monoidal functor from $(\cC,\bullet)$ to $(\cD,\star)$ is a functor $F$ from $\cC$ to $\cD$ such that $F(\bullet)=\star$, and the $n$-category of $E_0$-monoidal functors is denoted by $\Fun^{E_0}((\cC,\bullet),(\cD,\star)).$ 
    
    Let $\id^0_\bullet=\bullet,$ and iteratively, $  \id^{k+1}_\bullet:=\id_{\id^k_\bullet}$. 
    For $k\geq1$, the $k$-th looping of $(\cC,\bullet)$ is defined to be $\Omega^k\cC:=\Hom_\cC(\id^{k-1}_\bullet,\id^{k-1}_\bullet).$ We take the convention that the $0$th looping $\Omega^0\cC:=\cC$.
    To ease the notation, in the following we use the same symbol $\bullet$ for the distinguished object or just omit it and write $\cC$ when no confusion arises.
    
    Denote by $\Omega^k\cC|_\bullet$ the (Karoubi complete) full subcategory of $\Omega^k\cC$ generated by the distinguished object $\id^k_\bullet$ via direct sum and condensation.
    For $k\geq 1$, an $E_k$-multi-fusion $n$-category is a pair $(\cB,\cP)_n^k$ where $\cB$ is an $(n+1)$-vector space and $\cP$ is an $E_0$-multi-fusion $(n+k)$-category satisfying $\cB=\Omega^k\cP, \Omega^l\cP=\Omega^l\cP|_\bullet,\ \forall 0\leq l<k$. The $E_k$-monoidal functors is defined to be $\Fun^{E_k}((\cB,\cP)_n^k,(\cB',\cP')_n^k):=\Fun^{E_0}(\cP,\cP')$. $(\cB,\cP)_n^k$ is called fusion when $\bullet$, or equivalently $\id_\bullet^l$ for all $0\leq l\leq n+k$, is simple. Again, by abusing notation, we may simply write $\cB$ instead of $(\cB,\cP)_n^k$. We may also drop the prefix $E_k$ when $k=1$.  $\id^k_\bullet$ is called the tensor unit of $\cB$ and may be denoted by $\one_\cB$. There is an obvious forgetful functor $(\cB,\cP)_n^k\mapsto (\cB,\Omega^{k-l}\cP)_n^l$ allowing us to view an $E_k$-multi-fusion $n$-category as an $E_l$-multi-fusion $n$-category for $0\leq l <k$. 
    
     For $k\geq1$, the delooping  (or suspension) of an $E_k$-multi-fusion $n$-category $(\cB,\cP)_n^k$, denoted by $\Sigma(\cB,\cP)_n^k$, or simply $\Sigma\cB$,
    is the $E_{k-1}$-multi-fusion $(n+1)$-category $(\Omega^{k-1}\cP,\cP)_{n+1}^{k-1}$. Its looping, denoted by $\Omega(\cB,\cP)_n^k$ or simply $\Omega\cB$, is the $E_{k+1}$-multi-fusion $(n-1)$-category $(\Omega^{k+1}\cP,\cP|_{\id^k_\bullet})_{n-1}^{k+1}$, where $\cP|_{\id^k_\bullet}$ denotes the subcategory of $\cP$ by dropping $k$-morphisms not generated by $\id^k_\bullet$ so that $\Omega^k (\cP|_{\id^k_\bullet})=\Omega^k \cP|_\bullet$. By definition, an $E_k$-multi-fusion $n$-category $(\cB,\cP)_n^k$ satisfies $\cP=\Sigma^k \cB$.
\end{definition}

\begin{remark}
    The above definition assumes that the higher vector spaces are already constructed. Alternatively, one can construct $\Sigma \cB$ via $\mathrm{Kar}(B\cB)$, i.e., take the Karoubi-completion of the one-point delooping of $\cB$ ($B\cB$ is a higher category with only one object $\bullet$ and $\Hom_{B\cB}(\bullet,\bullet)=\cB$). $\Sigma\cB = \mathrm{Kar}(B\cB)$ is also called the condensation completion of $\cB$.
\end{remark}

\begin{proposition} $\Omega$ is the left inverse of $\Sigma$: $\Omega\Sigma\cB=\cB$. $\Sigma$ is left adjoint to $\Omega$:
       $$\Fun^{E_{k-1}}(\Sigma\cB,\cC)=\Fun^{E_k}(\cB,\Omega\cC).$$
\end{proposition}
\begin{proof}
     $\Omega\Sigma\cB=\cB$ is obvious. By definition $\Fun^{E_{k-1}}(\Sigma(\cB,\cP)_n^k, (\cC,\cQ)_{n+1}^{k-1})=\Fun^{E_0}(\cP,\cQ)$. An $E_0$-monoidal functor $F$ preserves the distinguished object, thus $F(\Omega^l\cP)=F(\Omega^l\cP|_\bullet)$ falls in $\Omega^l\cQ|_\bullet$ for $1\leq l\leq k-1$, and we have the same functor category after replacing $Q$ for the subcategory $\cQ|_{\id^{k-1}_\bullet}$:  
     \begin{align*} &\Fun^{E_{k-1}}(\Sigma(\cB,\cP)_n^k, (\cC,\cQ)_{n+1}^{k-1})\\
      &=\Fun^{E_0}(\cP,\cQ)=\Fun^{E_0}(\cP,\cQ|_{\id^{k-1}_\bullet})
     \\&=\Fun^{E_k}((\Omega^k\cP,\cP)_n^k,(\Omega^k\cQ,\cQ|_{\id^{k-1}_\bullet})_n^k)\\&
     =\Fun^{E_k}((\cB,\cP)_n^k,\Omega(\cC,\cQ)_{n+1}^{k-1}).
     \end{align*}
\end{proof}
\begin{remark}
  Given an $(n+1)$-vector space $\cX$, and an object $x\in \cX$, in general $\Sigma\Omega(\cX,x) \neq \cX$. It is a technical convenient assumption, that when $\cX$ is indecomposable (not the direct sum of two non-zero $(n+1)$-vector spaces), for any non-zero object $x\in \cX$ one has $\Sigma\Omega(\cX,x)=\cX$. For a generic $\cX$, suppose $\cX=\oplus_i \cX_i$ where each $\cX_i$ is indecomposable, then one can choose non-zero objects $x_i\in\cX_i$ for each $i$, and $\Sigma\Omega(\cX,\oplus_i x_i)=\cX$.
\end{remark}
Following~\cite{KZ2011.02859}, given an $E_k$-multi-fusion $n$-category $(\cB,\cP)_n^k$ and $-k+1\leq l\leq 0$, we denote by $\cB^{\op l}$ the $E_k$-multi-fusion $n$-category $(\cB,\cP^{\op (k+l)})_n^k$. In particular $\cB^\rev:=\cB^{\op0}$, $\ov\cB:=\cB^{\op(-1)}$, meaning reversing the tensor product and braiding respectively. Using delooping, $\Sigma^k(\cB^{\op l})=(\Sigma \cB)^{\op (k+l)},$ in particular, $\Sigma (\cB^\rev)=(\Sigma\cB)^\op$, $\Sigma^2 \ov\cB= (\Sigma^2\cB)^\op=\Sigma((\Sigma\cB)^\rev).$
\begin{remark}
    $(\ve[n],(\ve[n+k],\bullet=\ve[ n+k-1 ]) )_n^k$ is an $E_k$-fusion $n$-category for arbitrary $k$. Indeed, $\ve[n]$ is an $E_\infty$-fusion (symmetric) $n$-category. We also have $\Sigma (\ve[n])=\ve[n+1]$, $\ve[n]=\Sigma^n\C$, $\ve[n]^{\op l}=\ve[n]$ and $\Omega(\cC\bt\cD)=\Omega\cC\bt\Omega\cD, \ \Sigma(\cC\bt\cD)=\Sigma\cC\bt\Sigma\cD$. Moreover, for any symmetric fusion category, such as $\Rep G$, or $\sve$ the category of super vector spaces, we can consider its iterated delooping, denoted by $\Rep[n]G:=\Sigma^{n-1}\Rep G$ or $\sve[n]:=\Sigma^{n-1}\sve$. 
\end{remark}
\begin{definition}
    [$E_k$-center]
    The $E_0$-center of an $E_0$-multi-fusion $n$-category $(\cC,\bullet)$ is defined to be $Z_0(\cC):=\Fun(\cC,\cC)=\Omega( \ve[n+1],\cC )$, which is a multi-fusion $n$-category. The $E_k$-center of an $E_k$-multi-fusion $n$-category $(\cB,\cP)_n^k$ is the $E_{k+1}$-multi-fusion $n$-category $Z_k(\cB):=\Omega^k Z_0(\cP)=\Omega^k Z_0(\Sigma^k\cB)=\Omega^k \Fun(\Sigma^k\cB,\Sigma^k\cB)$.
\end{definition}
\begin{remark}\label{rem.center}
    One physical interpretation of center is to compute the bulk. More precisely, an object $\cX\in \ve[n+1]$ is a dualizable Karoubi complete $\ve[n]$-module. Without losing generality, we may suppose that $\cX$ is indecomposible.
    Physically, $\cX$ labels an anomaly-free $n$+1D topological order, and every object $x\in \cX\cong \Fun(\ve[n],\cX)$ labels a gapped boundary condition of $\cX$.
    The boundary-bulk relation is automatically satisfied: For any object (boundary condition) $x\in \cX$, we have
    \begin{align}
    &\Omega Z_0(\cX)=\Omega Z_0(\Sigma\Hom_\cX(x,x))\nonumber\\
    &=Z_1(\Hom_\cX(x,x))=Z_1(\Omega(\cX,x)).
    \label{eq.bb}
    \end{align}
   where $Z_0(\cX)\equiv \Fun(\cX,\cX)$ is the $E_1$-multi-fusion $n$-category of codimension-1 and higher defects of the $n$+1D bulk.,  $\Omega Z_0(\cX)$ is the $E_2$-multi-fusion $(n-1)$-category of codimension-$2$ and higher defects in the $n$+1D bulk, and $\Hom_\cX(x,x)$ is $E_1$-multi-fusion $(n-1)$-category of codimension-1 defects on the $n$D boundary. We see that the $E_1$-center takes the defects on the boundary as input and computes the defects of the same dimension in the bulk as output. We can also verify that the bulk of bulk is trivial:  $Z_0(\cX)=\Fun(\cX,\cX)=\Omega(\ve[n+1],\cX)$,  $Z_1(Z_0(\cX)) = \Omega Z_0(\ve[n+1])=\ve[n]$.
    
    In short, an object in $\ve[n+1]$ describes a collection of potentially anomalous $n$D topological orders with the same bulk, or by boundary-bulk relation, an anomaly-free $n$+1D topological order with defects $\Fun(\cX,\cX)$ whose gappable boundary conditions are objects in $\cX$. 
    

More generally, let $\cB$ be an $E_k$-multi-fusion $(n-k)$-category ($n\geq k\geq 1$),  $\Sigma^k\cB\in \ve[n+1]$ labels an $n$+1D bulk. Then $\cB = \Omega^k\Sigma^k\cB=\Omega^{k-1}\Hom_{\Sigma^k\cB}(\bullet,\bullet)$ is the category of codimension-$k$ defects on the canonical boundary $\bullet$. By delooping of \eqref{eq.bb}, the $E_k$-center $Z_k(\cB)=\Omega^k Z_0(\Sigma^k\cB)$ computes the codimension-$(k+1)$ defects in the $n$+1D bulk, which are of the same spacetime dimension ($n+1-(k+1)=n-k$) as the codimension-$k$ defects on the $n$D boundary. We can as well think $Z_0(\cX)$ as the $E_0$ version of computing codimension-1 defects in the bulk, which are of the same dimension as the boundary conditions (objects in $\cX$).
\end{remark}
\begin{remark}\label{rem.nonchiral}
Following~\cite{KZ2107.03858}, we call a {multi-}fusion $n$-category $\cB$ non-chiral when $\cB  =Z_0(\cC)$ for some $(n+1)$-vector space $\cC$ . In the higher linear algebra picture, $\cB$ is non-chiral if and only if $\Sigma\cB\cong \ve[n+1]$. For one direction, if $\cB$ is non-chiral, then $\Sigma\cB= \Sigma Z_0(\cC)=\Sigma \Omega({\ve[n+1]},\cC) = \ve[n+1]$. For the other direction, if $\Sigma\cB\cong \ve[n+1]$, let $\cC\in \ve[n+1]$ corresponding to $\bullet\in \Sigma\cB$, then $\cB = \Omega\Sigma\cB = \Hom_{\Sigma\cB}(\bullet,\bullet)\cong \Hom_{\ve[n+1]}(\cC,\cC)=Z_0(\cC)$. Similarly, an (braided) $E_2$-fusion $n$-category $\cM$ is called non-chiral when $\cM=Z_1(\cC)$ {for some fusion $n$-category $\cC$} (note that this condition is stronger than zero chiral central charge for $n=1$), which is equivalent to $\Sigma^2\cM\cong \ve[ n+2 ]$. {On the one hand, if $\cM = Z_1(\cC)$ for some fusion $n$-category $\cC$, then $\Sigma \cM=\Sigma Z_1(\cC) = Z_0(\Sigma \cC)$. The second equal sign is from the fact that the condensation completion of a fusion $n$-category is indecomposable. From the above discussion, $\Sigma^2\cM = \ve[n+2]$ since $\Sigma\cM$ is non-chiral. On the other hand, if $\Sigma^2\cM \cong \ve[n+2]$, let $\cX\in \ve[n+2]$ be the image of $\bullet\in \Sigma^2\cM$ under the equivalence. $\cX$ is indecomposable since $\cM$ is of fusion type. Then, $\Sigma \cM=\Fun(\cX,\cX)$   and $\cM = \Omega\Fun(\cX,\cX) = Z_1(\Hom_{\cX}(x,x))$. $\Hom_{\cX}(x,x)$ is of fusion type if $x\in\cX$ is simple.} 
\end{remark}

    \begin{example}
    To describe the gapped boundary of $2$+1D toric code, one can take $\Rep[2]\mathbb Z_2 \in \ve[3]$. Objects in $\Rep[2]\mathbb Z_2$ can be considered as separable algebras in $\Rep\mathbb Z_2=\{\one,e\}$, there are two different kinds of algebras $A_1=\one, \ A_2=\one\oplus e$ corresponding to two gapped boundaries of toric code.  We then have $\Omega \Fun(\Rep[2]\mathbb Z_2, \Rep[2] \mathbb Z_2)= Z_1(\Rep \mathbb Z_2)$  as the modular tensor $1$-category of particle-like excitations in the bulk toric code, and 
    $\Hom_{\Rep[2]\mathbb Z_2}(A_i,A_i)\cong{}_{A_i}({\Rep{\mathbb Z_2}})_{A_i}\cong \Rep \mathbb Z_2$ (see for example \cite{LZ2305.12917}) as the fusion $1$-category of particle like excitations on the gapped boundary. They satisfy the boundary-bulk relation automatically.

    \end{example}
    \begin{example}
        Consider a modular tensor category $\cM$ describing the particles of an anomaly-free, but potentially chiral, $2$+1D topological order. As $\cM$ may not have a gapped boundary, we need to go one higher dimension and consider the $2$+1D topological order as the gapped boundary of the trivial $3$+1D bulk. Then we can take $\cX=\Sigma^2 \cM \in \ve[4]$, and the boundary condition $x=\bullet\in \cX$. Since the topological order is anomaly-free, we have $Z_1(\Sigma \cM) = \ve[2]$, then $Z_0(\cX)=Z_0(\Sigma ^2  \cM)=\Sigma Z_1(\Sigma \cM) = \ve[3]$, which is the correct fusion 3-category of codimension-1 defect in the trivial $3$+1D bulk, and $\Omega(\cX,x)=\Hom_{\Sigma^2 \cM}(\bullet,\bullet) = \Sigma\cM$ is the correct fusion 2-category of codimension-1 defects of the $2$+1D anomaly-free topological order. 
    \end{example}

\begin{definition}
    [Module category] Given a multi-fusion $n$-category $\cB$, we denote by $\lmd\cB{ \ve[n+1] }:=\Fun(\Sigma\cB,\ve[n+1])$ the $(n+1)$-category of (dualizable Karoubi complete) left $\cB$-module $n$-categories. Unpacking it, an object $F$ in $\lmd\cB{ \ve[n+1] }$ is an $(n+1)$-vector space $\cX:=F(\bullet)$ together with a monoidal functor $\Omega_\bullet\cF:\cB\to \Fun(\cX,\cX)=Z_0(\cX)$ where $\Omega_\bullet F$ denotes the restriction of $F$ to $\Omega(\Sigma\cB,\bullet)=\cB$. The hom-category between $F$ and $G$, with $\cX=F(\bullet),\cY=G(\bullet)$, is denoted by $\Fun_\cB(\cX,\cY):=\Hom_{\lmd\cB{ \ve[n+1] }}(F,G)$, which is the $n$-category of (higher) natural transformations, whose objects are by definition functors from $\cX$ to $\cY$ that commutes\footnote{The way how they commute needs to be specified by higher structures and coherence conditions, which we suppress here and assume that they are automatically addressed by the theory of functors and natural transformations of $n$-category.} with the action of $\cB$.
    For right module and bimodule, we denote by $$\rmd{\cB}{\ve[n+1]}:=\Fun(\Sigma(\cB^\rev),\ve[n+1]),$$ 
    \begin{align*}
    &\bmd{\cB}{\cC}{ \ve[n+1] }:=\Fun(\Sigma(\cB\bt\cC^\rev),\ve[n+1])
    \\&=\Fun(\Sigma\cB,\Fun(\Sigma(\cC^\rev),\ve[n+1])),
    \end{align*}
    and $\Fun_{\cB^\rev}(\cX,\cY),\Fun_{\cB|\cC}(\cX,\cY)$ are similarly understood.
\end{definition}
\begin{remark}\label{rmk.ew}
By the properties of higher vector spaces, we know $$\rmd{\cB}{\ve[n+1]}={\Fun((\Sigma\cB)^\op,\ve[n+1])}=\Sigma\cB,$$ $$\lmd\cB{ \ve[n+1] }=(\Sigma\cB)^\op,$$ and $$\bmd{\cB}{\cC}{ \ve[n+1] }=\Fun(\Sigma\cB,\Sigma\cC)=\Sigma\cC\bt(\Sigma\cB)^\op.$$
Recall that the Yoneda embedding is an equivalence for higher vector spaces,
$$
\begin{array}{rcl}
  \Sigma\cB &\to & \Fun((\Sigma\cB)^\op,\ve[n+1])\\
  X &\mapsto &\Hom_{\Sigma\cB}(-,X)
\end{array}
$$
and by the general theory of condensation~\cite{GJ1905.09566}, objects $X,Y\in \Sigma\cB$ can be identified with algebras (condensation monads) in $\cB$ (with $\bullet$ identified with the tensor unit as the trivial algebra) and $\Hom_{\Sigma\cB}(Y,X)$ can be identified with the category of $X$-$Y$-bimodules in $\cB$, $\Hom_{\Sigma\cB}(Y,X)=\bmd{X}{Y}{ \cB }$. The first two equivalences are thus the higher categorical analogs of the Eilenberg-Watts theorem:
$$\begin{array}{rcl}
  \Sigma\cB &\cong& \rmd{\cB}{ \ve[n+1] }\\
  X&\mapsto & \Hom_{\Sigma\cB}(\bullet,X)=\lmd{X}{ \cB }\\
  \Hom_{\Sigma\cB}(Y,X)=\bmd{X}{Y}{ \cB }&\cong &\Fun_{\cB^\rev}(\lmd{ Y }{ \cB },\lmd{ X }{ \cB })
\end{array}
$$
$$\begin{array}{rcl}
  (\Sigma\cB)^\op &\cong& \lmd\cB{ \ve[n+1] }\\
  X&\mapsto & \Hom_{\Sigma\cB}(X,\bullet)=\rmd{X}{ \cB }\\
  \Hom_{\Sigma\cB}(Y,X)=\bmd{X}{Y}{ \cB }&\cong &\Fun_{\cB}(\rmd{ X }\cB,\rmd Y\cB)
\end{array}
$$
\end{remark}
\begin{remark}
The $E_0$-center of $\Sigma\cB$ is 
\begin{align*}
&\Fun(\Sigma\cB,\Sigma\cB)=\Fun(\Sigma\cB\bt(\Sigma\cB)^\op,\ve[n+1])
\\&={\bmd{\cB}{\cB}{ \ve[n+1] }^\rev},
\end{align*}
 with the distinguished object $\id_{\Sigma\cB}\in \Fun(\Sigma\cB,\Sigma\cB)$, or $\cB=\Hom_{\Sigma\cB}(\bullet,\bullet)$ as the regular $\cB$-$\cB$-bimodule. Therefore, $Z_1(\cB)=\Omega Z_0(\Sigma\cB)=\Fun_{\cB|\cB}(\cB,\cB).$
\end{remark}

\GY{
Let $\cB$ be a multi-fusion $n$-category, $\cM$ a right $\cB$-module and $\cN$ a left $\cB$-module. We can then define the relative tensor product of $\cM$ and $\cN$, denoted by $\cM\btr[\cB]\cN$. In Appendix~\ref{sec.tensor}, we review the rigorous definition of relative tensor product for $n=1$, and sketch the generalization for arbitrary $n$.
}

\begin{proposition}
    Relative tensor product of module categories is realized by the duality pairing in higher vector spaces.
\end{proposition}
\begin{proof}
     Note that $$\rmd{\cB}{ \ve[n+1] }=\Fun((\Sigma\cB)^\op,\ve[n+1])=\Sigma\cB,$$ and $$\lmd\cB{ \ve[n+1] }=\Fun(\Sigma\cB,\ve[n+1])=(\Sigma\cB)^\op$$ are dual to each other in $\ve[ n+2 ]$. As the Yoneda embeddings are equivalences, we can take two arbitrary objects $X,Y\in \Sigma\cB$ to represent two arbitrary functors $\Hom_{\Sigma\cB}(-,X)$ and $\Hom_{\Sigma\cB} (Y,-)$. Thus, the duality pairing is
     \begin{equation*}
         \begin{array}{rclcl}
             \Fun((\Sigma\cB)^\op,\ve[n+1])&\bt&\Fun(\Sigma\cB,\ve[n+1]) 
             &\cong& \Sigma\cB \bt \Fun(\Sigma\cB,\ve[n+1]) 
             \\
             \Hom_{\Sigma\cB}(-,X)&\bt&\Hom_{\Sigma\cB} (Y,-) & \mapsto& X\bt\Hom_{\Sigma\cB} (Y,-) \\
             &&&\xrightarrow{\text{pairing}}& \ve[n+1]\\
             &&&\mapsto& \Hom_{\Sigma\cB}(Y,X)
         \end{array}
     \end{equation*}
    By convention, the corresponding right $\cB$-module and left $\cB$-module are $\cX=\Hom_{\Sigma\cB}(\bullet,X)=\lmd X\cB$ and $\cY=\Hom_{\Sigma\cB}(Y,\bullet)=\rmd Y\cB$.
     Then, $$\Hom_{\Sigma\cB}(Y,X)=\bmd X Y \cB = \lmd X \cB \btr[\cB] \rmd Y \cB=\cX\btr[\cB]\cY.$$
\end{proof}
\begin{remark}
  Combining with Remark~\ref{rmk.ew}, we further have 
  \begin{align*}
  \Hom_{\Sigma\cB}(Y,X)&=\bmd X Y \cB = \lmd X \cB \btr[\cB] \rmd Y \cB
  \\&=\Fun_{\cB}(\rmd{ X }\cB,\rmd Y\cB)=\Fun_{\cB^\rev}(\lmd{ Y }\cB,\lmd X\cB).
  \end{align*}
  For a Morita invariant version (i.e., in terms of $\cB$-module categories instead of algebras inside $\cB$), let $\cN=\rmd X\cB=\Hom_{\Sigma\cB}(X,\bullet)$, $\cN'= \rmd Y\cB=\Hom_{\Sigma\cB}(Y,\bullet)$, we then have $\lmd X\cB=\Hom_{\Sigma\cB}(\bullet,X)=\Fun_{\cB}(\cN,\cB)$, and
  \[ \Fun_\cB(\cN,\cN')=\Fun_{\cB}(\cN,\cB)\btr[\cB] \cN'.\]
  Similarly for right $\cB$-modules $\cM,\cM'$,
  \[ \Fun_{\cB^\rev}(\cM,\cM')=\cM'\btr[\cB]\Fun_{\cB^\rev}(\cM,\cB).\]
\GY{For a right $\cB$-module $\cM$ and a left $\cB$-module $\cN$, the following formula can be used to explicitly characterize and compute the relative tensor product:
\begin{equation*}
\begin{split}
\cM\btr[\cB]\cN & =\Fun_{\cB}(\Fun_{\cB^{\rev}}(\cM,\cB),\cN)\\
    m\btr[\cB]n &= (G\in \Fun_{\cB^{\rev}}(\cM,\cB) \mapsto G(m)\triangleright n \in \cN)
\end{split}
\end{equation*}
where we use $\triangleright$ to denote the left $\cB$-action on $\cN$.}
\end{remark}

\begin{corollary}\label{cor.rt}
  More generally, we have the following commutative diagram:
  $$
  \begin{tikzcd}
    \Fun(\Sigma\cB,\Sigma\cC) \btr \Fun(\Sigma\cA,\Sigma\cB) \rar{\circ}\dar[equal] & \Fun(\Sigma\cA,\Sigma\cC)=\Fun(\Sigma\cB,\Sigma\cC) \btr[Z_0(\Sigma\cB)] \Fun(\Sigma\cA,\Sigma\cB)\dar[equal]\\
    \Sigma\cC\btr (\Sigma\cB)^\op\btr \Sigma\cB \btr (\Sigma\cA)^\op \rar{\text{pairing}}\dar[equal] & \Sigma\cC\btr (\Sigma\cA)^\op\dar[equal]\\
    \bmd\cA\cB{\ve[n+1]} \btr \bmd\cB\cC{\ve[n+1]} \rar{\btr[\cB]} & \bmd\cA\cC{\ve[n+1]}=\bmd\cA\cB{\ve[n+1]} \btr[\bmd\cB\cB{\ve[n+1]}] \bmd\cB\cC{\ve[n+1]} 
  \end{tikzcd}
  $$
  The first row is the composition of functors. These functors realize the relative tensor products on the right hand side.
\end{corollary}

\begin{definition}
    [Higher module]
    For an $E_{k+1}$-multi-fusion $n$-category $\cB$, a  (left) $E_{k}$-multi-fusion module over $\cB$ is an $E_{k}$-multi-fusion $n$-category $\cX$ equipped with an $E_{k+1}$-monoidal functor $\cB\to Z_{k}(\cX)$. In particular, for $k=1$, $E_1$-multi-fusion module is also referred to as multi-fusion module. A right $E_k$-multi-fusion module over $\cB$ is defined to be a left $E_k$-multi-fusion module over $\cB^{\op (-k)}$.
\end{definition}
\begin{remark}\label{rmk.hm}
    Since $\Sigma$ is left adjoint to $\Omega$, we have
    \begin{align*}
    \Fun^{E_{k+1}}(\cB,Z_{k}(\cX))&=\Fun^{E_{k+1}}(\cB,\Omega^{k} Z_{0}(\Sigma^{k}\cX))
    \\&=\Fun^{E_1}(\Sigma^{k}\cB,Z_{0}(\Sigma^k\cX)),
    \end{align*}
thus a left (right) $E_k$-multi-fusion module $\cX$ over $\cB$ is the same as a left (right) module $\Sigma^k\cX$ over $\Sigma^k \cB$, except that for $k=0$, an $E_0$-multi-fusion module $(\cX,\bullet)$ is a module $\cX$ with an additional distinguished object $\bullet$. The relative tensor product of higher modules is given by $\cX\btr[\cB]\cY=\Omega^k\Sigma^k(\cX\btr[\cB]\cY)=\Omega^k (\Sigma^k\cX\btr[\Sigma^k\cB]\Sigma^k\cY)$, and thus can also be realized via the duality pairing of the $(n+k+2)$-vector space $\Sigma^{k+1}\cB$.
\end{remark}
\begin{definition}[Relative center of bimodule, the center functor~\cite{KWZ1502.01690,KZ1507.00503,KZ2107.03858}]
  Let $\cA,\cB$ be multi-fusion $n$-categories and $\cX$ an $\cA$-$\cB$-bimodule. Note that the $E_1$-center of $\cB$ is given by $$Z_1(\cB)=\Fun_{\cB|\cB}(\cB,\cB)=\Omega(\bmd\cB\cB{\ve[n+1]}^\rev,\cB).$$ We consider more generally the looping $$Z_{\cA|\cB}(\cX):=\Fun_{\cA|\cB}(\cX,\cX)=\Omega(\bmd\cA\cB{\ve[n+1]},\cX),$$ called the relative center of $\cX$.
\end{definition}
\begin{remark}
  Supposing that $\Sigma\cA$ and $\Sigma\cB$ are indecomposable, thus $\bmd\cA\cB{\ve[n+1]}=(\Sigma\cA)^\op\bt\Sigma\cB$ is also indecomposable. We have $\Sigma Z_{\cA|\cB}(\cX)=(\Sigma\cA)^\op\bt\Sigma\cB$, thus $Z_0(\Sigma Z_{\cA|\cB}(\cX))=Z_0((\Sigma\cA)^\op)\bt Z_0(\Sigma\cB)=Z_0(\Sigma\cA)^\rev\bt Z_0(\Sigma\cB)$. Therefore, $Z_1(Z_{\cA|\cB}(\cX))=\ov{Z_1(\cA)}\bt {Z_1(\cB)}$. In other words, $Z_{\cA|\cB}(\cX)$ is a multi-fusion $Z_1(\cB)$-$Z_1(\cA)$-bimodule, which is moreover \emph{closed} in the sense of~\cite{KZ2107.03858}.
\end{remark}
\begin{theorem}[Functoriality of relative center]\label{thm.fun}
  Let $\cA,\cB,\cC$ be multi-fusion $n$-categories such that $\Sigma\cA,\Sigma\cB,\Sigma\cC$ are indecomposable. Let $\cX$ be an $\cA$-$\cB$-bimodule, and $\cY$ a $\cB$-$\cC$-bimodule. We have\footnote{Relative center (looping in bimodule categories) reverses the orientation of relative tensor product, which is essentially due to $\Fun_{\cB}(\cB,\cB)=\cB^\rev,$ or $\Sigma\cB=\Fun(\Sigma(\cB^\rev),\ve[n+1])$. To avoid further confusion, we always indicate such orientation-reversing by writing $\btr[\ov{Z_1(-)}]$, while keep other symbols in place.}
  \begin{align*}
  Z_{\cA|\cB}(\cX)\btr[\ov{Z_1(\cB)}]Z_{\cB|\cC}(\cY)=Z_{\cA|\cC}(\cX\btr[\cB]\cY).
  \end{align*}
\end{theorem}
\begin{proof}
  $\Sigma\cA,\Sigma\cB,\Sigma\cC$ being indecomposable implies that $\bmd\cA\cB{\ve[n+1]}$, $\bmd\cB\cB{\ve[n+1]},\ \bmd\cB\cC{\ve[n+1]}$ are all indecomposable, therefore,
  \begin{align*}
    \Sigma Z_{\cA|\cB}(\cX)&=(\bmd\cA\cB{\ve[n+1]},\cX),\\
    Z_0(\Sigma\cB)^\rev=\Sigma \ov{Z_1(\cB)}= \Sigma \ov{Z_{\cB|\cB}(\cB)}&=(\bmd\cB\cB{\ve[n+1]},\cB),\\
    \Sigma Z_{\cB|\cC}(\cY)&=(\bmd\cB\cC{\ve[n+1]},\cY).
  \end{align*}
  By Corollary~\ref{cor.rt}
  \begin{align*}
   &(\bmd\cA\cB{\ve[n+1]},\cX)\btr[{(\bmd\cB\cB{\ve[n+1]},\cB)}](\bmd\cB\cC{\ve[n+1]},\cY)
   \\&=(\bmd\cA\cC{\ve[n+1]},\cX\btr[\cB]\cY). 
  \end{align*}
  Then using Remark~\ref{rmk.hm}, we have the functoriality of the relative center   
  \begin{align*}
  &Z_{\cA|\cB}(\cX)\btr[\ov{Z_1(\cB)}]Z_{\cB|\cC}(\cY)=\Omega(\Sigma Z_{\cA|\cB}(\cX)\btr[\Sigma \ov{Z_1(\cB)}]\Sigma Z_{\cB|\cC}(\cY))\\
  &=\Omega\left(\bmd\cA\cC{\ve[n+1]}, \cX\btr[\cB]\cY\right)=Z_{\cA|\cC}(\cX\btr[\cB]\cY).
  \end{align*}
\end{proof}

\begin{definition}
  Two multi-fusion $n$-categories $\cA$ and $\cB$ are called Morita equivalent if $\Sigma\cA\cong \Sigma\cB$.
\end{definition}
\begin{remark}\label{rem.mor}
  Supposing that $\cA$ and $\cB$ are Morita equivalent, a direct consequence is that they have the same $E_1$-center $Z_1(\cA)=\Omega Z_0(\Sigma\cA)\cong \Omega Z_0(\Sigma\cB)=Z_1(\cB).$ Denote the images of distinguished objects by $X,Y$ under the equivalence $\Sigma\cA\cong\Sigma\cB$:
  $$\begin{array}{rcl}
    \Sigma\cA &\cong &\Sigma \cB, \\
    \bullet &\leftrightarrow& X,\\
    Y &\leftrightarrow& \bullet.
  \end{array}$$
  We then have
  \begin{align*}
  \cA&=\Hom_{\Sigma\cA}(\bullet,\bullet)\cong \Hom_{\Sigma\cB}(X,X),\\
  \cB&=\Hom_{\Sigma\cB}(\bullet,\bullet)\cong \Hom_{\Sigma\cA}(Y,Y).
\end{align*}
  In terms of the $\cA$-$\cB$-bimodule $\cX:=\Hom_{\Sigma\cB}(\bullet,X)\cong\Hom_{\Sigma\cA}(Y,\bullet)$, and $\cB$-$\cA$-bimodule $\cX^\vee:= \Fun_{\cB^\rev}(\cX,\cB)=\Hom_{\Sigma\cB}(X,\bullet)\cong \Hom_{\Sigma\cA}(\bullet,Y)=\Fun_{\cA}(\cX,\cA)$\footnote{$\Fun_{\cB^\rev}(\cX,\cB)$ is left dual to $\cX$ while $\Fun_\cA(\cX,\cA)$ is right dual to $\cX$. See Proposition 2.3 in~\cite{KZ2107.03858}. In general the left and right duals can be different. Here since $\cX$ is invertible, the left and right duals coincide and we use the same notation $\cX^\vee$.},
  the above means 
  \begin{align*}
    \cA&\cong \Fun_{\cB^\rev}(\cX,\cX)\cong \cX\btr[\cB]\cX^\vee,\\
    \cB&\cong \Fun_{\cA^\rev}(\cX^\vee,\cX^\vee)\cong \cX^\vee\btr[\cA]\cX.
  \end{align*} 
  Together we know the $\cA$-$\cB$-bimodule $\cX$ and $\cB$-$\cA$-bimodule $\cX^\vee$ are inverse to each other. 

  Suppose that $\cB$ is a fusion $n$-category, then $\Sigma\cB$ is indecomposable, and for any object $X\in\Sigma\cB$, $\Sigma\cB=\Sigma\Hom_{\Sigma\cB}(X,X)$. In terms of $\cX:=\Hom_{\Sigma\cB}(\bullet,X)=\lmd X \cB$, the dual multi-fusion $n$-category $\cB_\cX^\vee:=\Fun_{\cB^\rev}(\cX,\cX) =\bmd X X \cB$ is Morita equivalent to $\cB$ with $\cX=\lmd X\cB$ and $\cX^\vee=\rmd X \cB$ being two invertible bimodules. Moreover, $$\cB\cong \Fun_{(\cB_\cX^\vee)^\rev}(\cX^\vee,\cX^\vee)=(\cB_\cX^\vee)^\vee_{\cX^\vee}.$$ In particular, if $\cX$ is an indecomposable right $\cB$-module, $\cB_\cX^\vee$ is again a fusion $n$-category.
\end{remark}

\section{Category of SET orders}
In this section we discuss how to extract information from the category of SET orders $\Fun(\Sigma\cT,\cX)$. One can see that they are all natural or canonical procedures.
\medskip

Let's first clarify the physical meanings of $\cT$.
\begin{definition}
  A (finite generalized) symmetry in $n$+1D is a fusion $n$-category $\cT$. As such symmetry is not associated with any concrete physical system, and only an abstract multiplication rule, we may also say an \emph{abstract} symmetry $\cT$ for emphasis.
\end{definition}

\begin{definition}
  [Local fusion $n$-category] A local fusion $n$-category~\cite{KLW+2005.14178} is a fusion $n$-category $\cC$ equipped with a linear monoidal functor $\cC\to \ve[n]$. A symmetry $\cT$ is called local if $\cT$ is a local fusion $n$-category.
\end{definition}
\begin{definition}
  [$k$-symmetry and $k$-form symmetry] A $k$-symmetry in $n$+1D is a fusion $n$-category of the form $\Sigma^k \cC$ for some non-trivial $E_{k+1}$-fusion $(n-k)$-category $\cC$. Physically, this means that $k$-symmetry is generated by topological operators of codimension $k+1$ or higher. If $\Omega\cC=\ve[ n-k-1 ]$ and simple objects in $\cC$ are all invertible, $\Sigma^k\cC$ is a called a $k$-form symmetry. $k$-form symmetry is generated by invertible topological operators at exactly codimension $k+1$. 
\end{definition}
\begin{example}
Local 0-form symmetry is the ordinary invertible global symmetry and is described by the fusion $n$-category $\ve[n]_G$, the $G$-graded $n$-vector spaces, for some group $G$. 
  For a 1-symmetry $\Sigma\cC$ in 2+1D, $\cC$ is a braided fusion $1$-category. For a fusion 1-category we automatically have $\Omega\cC=\C=\ve[0]$. If all simple objects in $\cC$ are invertible, then $\cC$ must be a braided fusion $1$-category $\ve_G^{\omega_3}$ for some abelian group $G$, and abelian 3-cocycle $\omega_3$. So a general 1-form symmetry in $2$+1D is $\Sigma(\ve_G^{\omega_3})$. It is local if and only if $\omega_3=0$, and $\Sigma(\ve_G)=\Rep[2] G$.
  This matches the general result of local $1$-symmetry in $2$+1D; see also Example~\ref{2+1D1-sym}.
  
\end{example}

Next, we clarify the physical meaning of $\cX$.
\begin{definition}
  A (potentially anomalous) bare theory in $n$+1D is a (multi-)fusion $n$-category $\cB$. It physically corresponds to all the topological operators in an $n$+1D (potentially anomalous) topological ordered phases, or more generally the topological skeleton of an $n$+1D (potentially anomalous) quantum liquid~\cite{KZ2011.02859}. The term \emph{bare} emphasizes that we do not consider symmetry of $\cB$ yet. $\cB$ is called anomaly-free if $Z_1(\cB)=\ve[n]$.
\end{definition}

\begin{remark}
In most cases we will assume that $\cB$ is fusion, so that the symmetry is not spontaneously broken; see Example~\ref{eg.ssb}.
\end{remark}

\begin{definition}
 A symmetry assignment is a monoidal linear $n$-functor $\phi:\cT\to\cB$ from an abstract symmetry $\cT$ to a bare theory $\cB$.
\end{definition}
\begin{remark}
  It is very important to separate the notions of abstract symmetry, the bare theory, and the symmetry assignment, as different symmetry assignments really correspond to different physical systems. One may consider a simple example, a two-dimensional lattice system of spin 1/2's, equipped with either $180^\circ$ spatial rotation or global spin flips. In both case the bare theory and the abstract symmetry ($\Z_2$) are the same; they can be distinguished only by the symmetry assignment.
\end{remark}

\begin{definition}
  An anomaly category (of bare theories in $n$+1D) is a separable $( n+1 )$-category $\cX$ (an $(n+2)$-vector space $\cX\in\ve[n+2]$). A bare theory (of fusion type) is obtained by choosing a (simple) object $x\in \cX$ and taking $\cB:=\Omega(\cX,x)$.
\end{definition}

\GY{
\begin{remark}
    The anomaly category $\cX$ only captures the anomaly of the bare theory, we need to differ it from the 't Hooft anomaly of the symmetry (see Remarks \ref{rmk.sym1} and \ref{rmk.sym2}). 
\end{remark}
}

\begin{remark}\label{rem.anomaly}
  It is meaningful to consider the following types of $\cX$:
  \begin{itemize}
    \item For $\cX=\ve[n+1]$, $\cB:=\Omega(\ve[n+1],\cY)=\Fun(\cY,\cY)$ correspond to anomaly-free $n$+1D topological orders admitting gapped boundaries, i.e., $\cB$ is non-chiral in the sense of Remark~\ref{rem.nonchiral}.
    \item For invertible $\cX$, i.e., $\Fun(\cX,\cX)=\cX\bt\cX^\op=\ve[n+1]$, by the results in~\cite{KZ2107.03858,Joh2003.06663}, we have the following assertions: \begin{enumerate}
      \item $\cX$ is indecomposable, and for any $x\in \cX$, $\cX=\Sigma\Omega(\cX,x)$. 
      \item Take a simple $x\in \cX$, then $\cB:=\Omega(\cX,x)$ is a fusion $n$-category with trivial center $Z_1(\cB)=\ve[n]$. Moreover $\cB=\Sigma\Omega\cB$ and $\Omega\cB$ is a braided fusion $( n-1 )$-category with trivial center $Z_2(\Omega\cB)=\ve[n-1]$.
      \item The above combined, for any simple object $x\in \cX$, $\cX=\Sigma^2\Omega^2(\cX,x)$.
      \item The Morita-class of fusion $n$-categories with trivial center, Witt-class of braided fusion $( n-1 )$-categories with trivial center, and equivalence class of invertible separable $( n+1 )$-categories, one-to-one correspond to each other via looping and delooping: $$\cB\overset{\text{Morita}}{\simeq}\cB' \Leftrightarrow \Omega\cB\overset{\text{Witt}}{\simeq}\Omega\cB'\Leftrightarrow \Sigma\cB\simeq\Sigma\cB'.$$
    \end{enumerate}
    Therefore, an invertible $\cX$ collects Morita-equivalent fusion $n$-categories with trivial center, or Witt-equivalent braided fusion $( n-1 )$-categories with trivial center. These bare theories are anomaly-free (in the world of higher vector spaces). They can still possess other anomalies, such as that witnessed by invertible topological phases, or non-zero chiral central charges (framing anomaly).
    \item For non-invertible indecomposable $\cX$, bare theories $\Omega(\cX,x)$ necessarily possess the gravitational anomaly witnessed by the nontrivial $n$+2D topological order $\Fun(\cX,\cX)$.
  \end{itemize}
\end{remark}

Now, we are ready to define SET orders:
\begin{definition}
  A $\cT$-SET order is a functor $F$ in $\Fun(\Sigma\cT,\cX)$.
  The corresponding bare theory of $F$ is $\Omega(\cX,F(\bullet))$.
  \begin{enumerate}
    \item $F$ is called anomaly-free if $\cX$ is invertible. 
    \item $F$ is called a symmetry protected topological (SPT) order if $\cX=\ve[n+1]$ and $F(\bullet)$ is invertible, which implies that the bare theory is trivial, i.e., $\Omega(\cX,F(\bullet))=\ve[n]$.
    \item The (higher) defects between SET orders are (higher) morphisms in $\Fun(\Sigma\cT,\cX)$.
    \item A symmetry assignment $\phi:\cT\to \cB$ upgrades to a functor $\Sigma\phi:\Sigma\cT\to\Sigma\cB$ and thus an anomaly-free $\cT$-SET order if $\cB$ is anomaly-free. 
    \item For a local symmetry $\cT$ , the local data $\cT\to\ve[n]$ defines a trivial $\cT$-SET (or $\cT$-SPT) phase. A boundary of a $\cT$-SET order $F$ is a defect between $F$ and the trivial $\cT$-SET order.
  \end{enumerate}
\end{definition}

\begin{remark}\label{rmk.sym1}
  In the literature a symmetry $\cT$ is called anomaly-free if it admits a trivial phase~\cite{TW1912.02817}. We see in our framework that this requirement exactly corresponds to that $\cT$ is local (admits a fiber functor to $\ve[n]$). Moreover, a symmetry $\cT$ being \emph{anomalous}, i.e., $\cT$ does not admit a monoidal functor $\cT\to \ve[n]$, is in fact a property of $\cT$. But a symmetry $\cT$ being local or anomaly-free, is a structure rather than a property, since the choice $\cT\to\ve[n]$ has to be made, which is not unique even for the usual group-like symmetry $\cT=\ve[n]_G$ (though in most cases we implicitly choose the forgetful functor $\ve[n]_G\to \ve[n]$ as the local structure). Because of these subtleties, when the fusion category $\cT$ is given, we avoid talking about its anomaly: $\cT$ is automatically anomalous if it does not admit any fiber functor; on the other hand, we will specify a local structure $\cT\to\ve[n]$ which means $\cT$ is already anomaly-free. By the anomaly of $\cT$-SET orders, we mainly mean the anomaly of the underlying topological orders. Only when $\cT$ is not explicitly given, for example when talking about $G$ symmetry instead of $\ve[n]_G$ symmetry, we may emphasize on ``anomalous $G$-action''; see Example~\ref{eg.3dset}.
\end{remark}
\GY{\begin{remark}\label{rmk.sym2}
    The anomaly of a symmetry $\cT$ is also called 't Hooft anomaly~\cite{hooft1980naturalness}. 't Hooft anomaly by definition is the obstruction to (ordinary) gauging. As we will discuss in the next section, gauging is to condense symmetry defects. The obstruction to gauging is that we cannot condense all possible symmetry defects. Mathematically it means that the symmetry $\cT$ does not have $\ve[n]$  as its module, or equivalently, the symmetry $\cT$ does not have a fiber functor to $\ve[n]$.
\end{remark}}
\begin{remark}
  It is not reasonable to consider defects between $F:\Sigma\cT\to \cX$ and $G:\Sigma\cT\to\cY$ for different $\cX$ and $\cY$. The anomalies $\cX$ and $\cY$ play a similar role as the symmetry $\cT$, as can be seen from the formula $\Fun(\Sigma\cT,\cX)=\cX\bt(\Sigma\cT)^\op$; they both constrain the (low-energy) dynamics of the physical system. To compare $F:\Sigma\cT\to \cX$ and $G:\Sigma\cT\to\cY$, probably one has to introduce explicit ``anomaly change'' $\cX\to\cY$ (Definition~\ref{def.eab}), in analogy to explicit symmetry breaking (Definition~\ref{def.esb}).
\end{remark}

It is also of practical use to consider ``anomaly-free'' relative to a fusion $n$-category $\cV$, when one thinks operator in $\cV$ as building blocks of the physical system and free from anomalies. We call such $\cV$ a \emph{background}. For example, when considering fermion systems, one should take $\cV=\sve[n]=\Sigma^{n-1}\sve$ as the background. In this paper we only give the following definition; related research may be done in future works.
\begin{definition}
  A symmetry $\cT$ is $\cV$-local if it is a $\cV$-local fusion category~\cite{KLW+2005.14178} i.e., equipped with a monoidal functor $\cT\to \cV$. An anomaly category $\cX$ is called anomaly-free relative to $\cV$ if $Z_0(\cX)=Z_0(\Sigma\cV)$. Equivalently, a bare theory $\cB$ (of fusion type) is called anomaly-free relative to $\cV$ if $Z_1(\cB)=Z_1(\cV)$.
\end{definition}

\begin{example}[0+1D $G$-SET]
  The only invertible object in $\ve[2]$ is $\ve$, and $\ve[0]_G=\C G$ the group algebra,
  \[  \Hom_{\ve[2]}(\Sigma(\C G),\ve)=\Fun(BG,\ve)=\Rep G.\]  
  A (0+1D) quantum system $(H,V)$ with symmetry means that the Hilbert space $V$ is an object in $\Rep G$, i.e., $V$ carries a group representation:
  \[ \rho: G\to \End{V}.\]
  and the Hamiltonian $H$ is a morphism in $\Rep G$, i.e., $H$ is a symmetric operator $\rho_g H=H\rho_g,\forall g\in G$. We can see here that $\Rep G$ ignores all the information related to the spatial dimensions; the quantum system is described as whole as if the spatial dimension is zero.
\end{example}

\begin{example}[1+1D $G$-SET]

  The only invertible object in $\ve[3]$ is $\ve[2]$
  \[ \Hom_{\ve[3]}(\Sigma \ve_ G,\ve[2])=\Rep[2] G.\]  
  Objects in $\Rep[2] G=\Sigma\Rep G$ are algebras in $\Rep G$, (whose Morita-classes are) parametrized by a subgroup $H\subset G$ together with a 2-cocycle $\omega_2\in H^2(H,U(1))$ (see, for example,~\cite{EGNO0301027,EO0301027,Ost0111139}), which agrees with the well-known classification of 1+1D phases with symmetry $G$ (unbroken symmetry $H$ which possibly hosts a SPT order $\omega_2$)~\cite{CGW1008.3745,CGW1103.3323,SPC1010.3732}. 

\end{example}

\begin{example}\label{eg.3dset}
  Consider the special case that $\cX=\ve[n+1]$. Note that
  \[ \Fun(\Sigma \ve[n]_G,\ve[n+1])\cong{}_{\ve[n]_G}\ve[n+1]\]
  can be identified with the $(n+1)$-category of separable left $\ve[n]_G$-module $n$-categories. For $F_1,F_2\in {}_{\ve[n]_G}\ve[n+1]$, defects between them are just $\ve[n]_G$-module functors in $\Fun_{\ve[n]_G}(F_1,F_2)$. 
  
  As a special case, we want to show that the so-called surface topological orders with anomalous $G$-action~\cite{CBVF1403.6491}, characterized by $\omega_4\in H^4(G,U(1))$ and monoidal functors $\ve[2]_G^{\omega_4}\to \Sigma\cC$, are just boundaries of 3+1D $G$-SPT orders.  Consider the category of 3+1D $G$-SET orders
  $$ \Fun(\Sigma \ve[3]_G,4\ve)\cong{}_{\ve[3]_G}4\ve.$$
  \begin{itemize}
    \item The trivial 3+1D $G$-SPT order is given by the forgetful functor $\ve[3]_G\to\ve[3]$, corresponding the left $\ve[3]_G$-module $(\ve[3]_G)_{\ve[2]_G}\cong \ve[3]$.
    \item Take the 3+1D $G$-SPT order, for some $\omega_4\in H^4(G,U(1))$ and MTC $\cC$, corresponding to the left $\ve[3]_G$-module $(\ve[3]_G)_{\ve[2]_G^{\omega_4}\bt\Sigma\ov\cC}\cong (\ve[3])_{\Sigma\ov\cC}=\Sigma^2\ov\cC=(\Sigma^2\cC)^\op$. 
      \item Defects between these two then form  $\Fun_{\ve[3]_G}((\ve[3]_G)_{\ve[2]_G^{\omega_4}\bt\Sigma\ov\cC},(\ve[3]_G)_{\ve[2]_G})\cong {}_{\ve[2]_G^{\omega_4}\bt\Sigma\ov\cC}(\ve[3]_G)_{\ve[2]_G}\cong{}_{\ve[2]_G^{\omega_4}\bt\Sigma\ov\cC}(\ve[3])\cong \Fun(\Sigma\ve[2]_G^{\omega_4}\bt(\Sigma^2\cC)^\op,\ve[3])\cong \Fun(\Sigma \ve[2]_G^{\omega_4}, \Sigma^2\cC)$. 
  \end{itemize}
  Therefore, a 2+1D topological order $\cC$ with anomalous $G$-action, characterized by a monoidal 2-functor $ \phi: \ve[2]_G^{\omega_4} \to \Sigma\cC,$
  corresponds to a boundary of the 3+1D $G$-SPT $(\ve[3]_G)_{\ve[2]_G^{\omega_4}\bt\Sigma\ov\cC}$, i.e., $\Sigma\phi\in\Fun(\Sigma \ve[2]_G^{\omega_4}, \Sigma^2\cC)$. Such anomalous $G$-action is a straightforward generalization to the anomaly-free $G$-actions $\ve[2]_G \to \Sigma\cC.$ In this example, one can also call $\Sigma\phi\in\Fun(\Sigma \ve[2]_G^{\omega_4}, \Sigma^2\cC)$ as an anomaly-free ($\Sigma^2\cC$ is invertible) $\ve[2]_G^{\omega_4}$-SET order, and clearly $\ve[2]_G^{\omega_4}$ is an anomalous symmetry for nontrivial $\omega_4$.
\end{example}

Next we discuss how to extract other information from the category of SET orders.
\begin{definition}
  Given a $\cT$-SET $F\in\Fun(\Sigma\cT,\cX)$, the corresponding \emph{charge category} is defined to be defects in $F$, i.e. $\Omega(\Fun(\Sigma\cT,\cX),F)$. In particular, if $F$ is obtained from a symmetry assignment $\phi:\cT\to\cB$, $F=\Sigma\phi: \Sigma\cT\to\Sigma\cB$, recall that $\Fun(\Sigma\cT,\Sigma\cB)=\bmd\cT\cB{\ve[n+2]}$ is identified with $\cT$-$\cB$-bimodules, and $\Sigma\phi$ is identified with the $\cT$-$\cB$-bimodule $\cB$ where the left action is induced by $\phi$. We denote this bimodule by $_\phi \cB$ for short (bimodules $_\phi \cB_\phi$ and $\cB_\phi$ are understood similarly.) The charge category of $\Sigma\phi$ is exactly the relative center of the $\cT$-$\cB$-bimodule $_\phi \cB$, denoted by $Z(\phi):=Z_{\cT|\cB}(_\phi\cB)=\Fun_{\cT|\cB}(_\phi\cB,_\phi\cB)$.
\end{definition}

\begin{remark}\label{rem.charge}
  Physically, the charge category consists of operators that are invariant under the symmetry transformations. To see this, note that forgetting the left module structure gives $\Fun_{\cT|\cB}(_\phi\cB,_\phi\cB)\to \Fun_{\cB^\rev}(\cB,\cB)\cong \cB$. Therefore, the bimodule functors in $Z(\phi)=\Fun_{\cT|\cB}(_\phi\cB,_\phi\cB)$ can be roughly thought as operators in $\cB$ that further commutes with the image of $\phi$:
  \[ Z(\phi)\sim \{U\in \cB|U\ot \phi(X) \sim \phi(X)\ot U, \forall X\in \cT\}.\]
\end{remark}

\begin{example}\label{eg.ssb}
    Let $\cT=\ve[n]_G$ and $\phi_1:\ve[n]_G\to \ve[n]$ the forgetful functor. The corresponding charge category is $Z(\phi_1)=\Rep[n] G=\Sigma^{n-1}\Rep G$, the iterative condensation completion of the usual representation category\footnote{This result means, in particular, that $\ve[n]_G$ is Morita-equivalent to $\Rep[n]G$, thus $\Rep[n+1]G=\Sigma \Rep[n]G=\Sigma\ve[n]_G$.}. In other words, the (higher-dimensional) charges are the usual symmetry charges and their higher dimensional condensates. In this case the symmetry $\cT$ is not broken. 
    Take another symmetry assignment $\phi_2:\ve[n]_G\to \Fun(\ve[n]_G,\ve[n]_G),\ g\mapsto g\ot -$. The corresponding charge category is $Z(\phi_2)=\ve[n]_G$. In this case the bare theory $ \Fun(\ve[n]_G,\ve[n]_G)$ has $|G|$ ground state sectors, and the charges are symmetry defects. Therefore it describes the spontaneous symmetry breaking phase.
  These examples show that from the charge category $Z(\phi)$ one can see whether the symmetry $\cT$ is spontaneously broken.
\end{example}

\begin{definition}
  \label{def.esb}
  Given two symmetries $\cH$ and $\cT$, if there is a monoidal functor $\eta:\cH\to \cT$, we can force an \emph{explicit symmetry change} by pulling back a $\cT$-SET to a $\cH$-SET, i.e., pre-composing $\Sigma\eta$, $F\mapsto F\circ\Sigma\eta$. In particular, when $\eta$ is an embedding, this is an explicit symmetry breaking from $\cT$ to $\cH$.
\end{definition}
\begin{definition}
  \label{def.eab}
  Given two anomaly categories $\cX$ and $\cY$, and functor $\mu:\cX\to \cY$, we can force an \emph{explicit anomaly change} by pushing forward a $\cT$-SET, i.e., post-composing $\mu$, $F\mapsto \mu\circ F$.
\end{definition}

Consider two category of SET orders $\Fun(\Sigma\cT,\cX)$ and $\Fun(\Sigma\cH,\cY)$, one can naively stack them
\[ \Fun(\Sigma\cT,\cX)\bt\Fun(\Sigma\cH,\cY)=\Fun(\Sigma\cT\bt\Sigma\cH,\cX\bt\cY).\]
However, both the symmetry and anomaly can change. In order to reduce the symmetry or anomaly to the original one, additional structures are required.
\begin{definition}
  A comultiplication of symmetry $\cT$ is a monoidal functor $\Delta: \cT\to \cT\bt\cT$ which is coassociative. Given two SET orders $F_1\in \Fun(\Sigma\cT,\cX),\ F_2\in \Fun(\Sigma\cT,\cY)$, the following defines their stacking which preserves the symmetry $\cT$
   \begin{gather*}
    (F_1\bt F_2)\circ \Sigma\Delta\in \Fun(\Sigma\cT,\cX\bt\cY).
  \end{gather*}
\end{definition}
\begin{remark}
  Preserving the anomaly is more tricky. One may tentatively use an associative functor $\cX\bt\cX\to\cX$ to preserve anomaly. However, it seems more practical to require $Z_0(\cX)$, $Z_0(\cY)$ and $Z_0(\cX\bt\cY)$ to be the same. Therefore, when $\cX$ and $\cY$ are both invertible, they are already good enough (for boson systems). For fermion systems ($\cV=\sve[n]$) or more generally any background $\cV$ which is an $E_k$-fusion $n$-category where $k\geq 2$, we know that the tensor functor of $\cV$, $\cV\bt \cV\xrightarrow{\ot}\cV$, is $E_{k-1}$-monoidal, and we can use the relative center $Z_{\cV|\cV\bt\cV}(\cV_\ot)$ to construct a stacking preserving anomaly: let $\cX$ and $\cY$ be both anomaly-free relative to $\cV$, $Z_0(\cX)=Z_0(\cY)=Z_0(\Sigma\cV)$, we define the staking preserving the background $\cV$ as the composition $\Sigma\cT_1\bt\Sigma\cT_2\xrightarrow{F_1\bt F_2}\cX\bt\cY\to \Sigma Z_{\cV|\cV\bt\cV}(\cV_\ot)\btr[\Sigma\ov{Z_1(\cV\bt \cV)}] (\cX\bt\cY)$.
\end{remark}
  
\begin{example}
  Group-like symmetries admits a natural comultiplication $\Delta(g)=(g,g)$. Therefore, in our framework there is a natural definition of the staking of    bosonic or fermionic SET orders with group-like symmetries.
\end{example}

\begin{remark}
  [Holography] Focus on the category $\Fun(\Sigma\cT,\cX)$. When we interpret $\cT$ as a symmetry and $\cX$ as a category of bare theories, it is the appropriate representation category of SET orders. However, as in Section~\ref{sec.hl}, there is yet another interpretation that $\Fun(\Sigma\cT,\cX)$ is the category of defects between $\Sigma\cT$ and $\cX$, with $\Sigma\cT$ and $\cX$ both viewed as labels of topological phases (objects in $\ve[n+2]$). Assume that $\cX$ is invertible, then $\Fun(\Sigma\cT,\cX)$ is just the category of boundaries of $Z_0(\Sigma\cT)=\Sigma Z_1(\cT)$. $Z_1(\cT)$ is nothing but the so-called categorical symmetry~\cite{JW1905.13279,JW1912.13492,KLW+2005.14178,JW2106.02069}, SymTFT~\cite{ABE+2112.02092,BS2304.02660,BS2305.17159}, symmetry TO~\cite{CW2203.03596,CW2205.06244} or quantum currents~\cite{LZ2305.12917} in the literature. (When $\cT=\ve[n]_G$, $Z_1(\ve[n]_G)$ exactly corresponds to the $G$-gauge theory in $n$+2D.) Moreover, when a SET order $F\in \Fun(\Sigma\cT,\cX)$ is fixed, the corresponding charge category provides another boundary of the SymTFT, i.e. $Z_1(\Omega(\Fun(\Sigma\cT,\cX),F))=Z_1(\cT)$, and $\cT, Z_1(\cT), \Omega(\Fun(\Sigma\cT,\cX),F)$ recovers the ``sandwich''~\cite{FMT2209.07471}. We see that the idea of SymTFT, symmetry/TO correspondence, or topological symmetry is a natural consequence of the representation principle, which appears to a coincidence that the same mathematical structure $\Fun(\Sigma\cT,\cX)$ admits two different physical interpretations. We will come back to this point again when discussing the calculation of gauging.
\end{remark}

The defects between SET orders or topological orders are closely related to the theory of phase transitions. It is proposed in Refs.~\cite{CW2203.03596,CW2205.06244}, that a symmetry $\cT$ is encoded in the symmetry TO $Z_1(\cT)$ and we can study the phase transitions via the boundaries of the symmetry TO. 
For two fusion $1$-categories, they are Morita equivalent if and only if their $E_1$ centers are braided equivalent~\cite{etingof2010fusion}. However this is not true for fusion $n$-category ($n\geq 2$)~\cite{Dec2208.08722,Dec2211.04917}, which means that when two symmetries $\cT$ and $\cH$ share the same symmetry TO, $Z_1(\cT)\cong Z_1(\cH)$, $\cT$ may not be Morita equivalent to $\cH$. Our following theorem further clarifies the relation between center and Morita classes:
\begin{theorem}
  Given two $(n+1)$-vector spaces $\cA$ and $\cB$, their $E_0$-centers are equivalent as multi-fusion $n$-categories, if and only if there exists an invertible object $\cX\in \ve[n+1]$ such that $\cA\bt\cX=\cB$.
\end{theorem}
\begin{proof}
  One direction is clear: $Z_0(\cB)=Z_0(\cA\bt\cX)=Z_0(\cA)\bt Z_0(\cX)= Z_0(\cA).$ For the other direction, supposing that $Z_0(\cA)\cong Z_0(\cB)$ is a monoidal equivalence, then $\cA^\op$ is a right $Z_0(\cB)$-module. Let $\cX=\cA^\op \btr[Z_0(\cB)] \cB$, by Corollary 3.10 in Ref.~\cite{KZ2107.03858}, we know $Z_0(\cX)=\ve[n]$ and thus $\cX$ is invertible. Moreover, $\cA\bt \cX=\cA\bt\cA^\op \btr[Z_0(\cB)] \cB=Z_0(\cA)\btr[Z_0(\cB)] \cB=\cB$.
\end{proof}
\begin{corollary} The following assertions are equivalent
  \begin{enumerate}
    \item Two fusion $n$-categories $\cT$ and $\cH$ have the same $E_1$-center, $Z_1(\cT)\cong Z_1(\cH)$.
    \item $\cT$ and $\cH$ are Morita equivalent up to an invertible object in $\ve[n+2]$, i.e. there exists invertible $\cX\in \ve[n+2]$ such that $\Sigma\cT\cong \Sigma\cH\bt\cX$.
    \item There exists a braided fusion $(n-1)$-category $\cM$ with $Z_2(\cM)=\ve[n-1]$, such that $\cT$ is Morita equivalent to $\cH\bt \Sigma\cM.$
  \end{enumerate}
\end{corollary}
\begin{remark}
The above results are straightforward generalization of those for fusion 1-categories~\cite{etingof2010fusion} and fusion 2-categories~\cite{Dec2211.04917}. An equivalent result was obtained in the higher condensation theory; see Theorem\textsuperscript{ph} 3.2.13 in Ref.~\cite{KZZZ2403.07813}.
\end{remark}
Physically, we can see that the category of SET orders, $\Fun(\Sigma\cT,\cX)=(\Sigma\cT)^\op\bt\cX$, is a refined description of the SymTFT $Z_1(\cT)=\Omega Z_0(\Sigma\cT)$, as the possible ambiguity of invertible $\cX$ has been explicitly considered.
Based on Ref.~\cite{LZ2305.12917} and the discussion about gauging or condensation in this paper, we believe that, between two Morita equivalent symmetries, there always exist continuous phase transitions which can be viewed as (generalized) spontaneous symmetry breaking. For two symmetries which are not Morita equivalent, it is reasonable to expect that phase transitions between them are physically distinct from those between Morita equivalent ones, or there is just no continuous phase transition. Such feature should be captured by the invertible higher vector space $\cX$.

Finally we discuss how the information of gauging is naturally encoded in the category of SET orders.
Note that objects in $\Sigma \cT$ can be identified with algebras in $\cT$. Given a functor $F: \Sigma \cT\to \cX$, we denote its restriction to the endo-category on object $A\in \Sigma \cT$, by $\Omega_A F:\Hom_{\Sigma\cT}(A,A)=\bmd A A \cT\to \Hom_\cX(F(A),F(A))$ which is a monoidal functor. When $A=\one_\cT$ the trivial algebra corresponding to $\bullet\in \Sigma\cT$, we get the symmetry assignment $\Omega_\bullet F:\cT\to  \Omega(\cX,F(\bullet))$ with bare theory $\cB:= \Omega(\cX,F(\bullet))$. For a nontrivial $A$, abusing $F(A)$ to denote the algebra in $\cB$, we get a monoidal functor $\Omega_A F:\bmd A A \cT\to \bmd{F(A)}{F(A)}{\cB}$.
\begin{definition}[Gauging]\label{def.gaugingalg}
 Given a symmetry assignment $\phi:\cT\to \cB$, and an (indecomposable) algebra $A\in\cT$ ($A$ labels the symmetry defects to be condensed, a datum necessary for gauging), the gauged theory is $\Omega_A \Sigma\phi: \bmd A A \cT\to \bmd{\phi(A)}{\phi(A)}{\cB}$, i.e., a bare theory $\bmd{\phi(A)}{\phi(A)}{\cB}$ with the gauge symmetry $\bmd A A \cT$.
\end{definition}
In short, gauging is just shifting the base object in $\Sigma\cT$. Since $\cT$ is a fusion $n$-category, $\Sigma\cT$ is indecomposable, and $F=\Sigma \Omega_A F: \Sigma\cT\to \cX$ for any $A\in \Sigma \cT$. Therefore, any gauged theory contains the same information as $\Omega_\bullet F: \cT\to\Omega(\cX,F(\bullet))$. In particular, the charge category is the same $Z(\Omega_A F)=Z(\Omega_\bullet F)$ for any $A\in \Sigma\cT$. In the next section we will discuss more details about gauging: we will give a Morita-invariant story (in terms of modules over $\cT$ instead of algebras in $\cT$), and give the explicit algorithm to calculate the gauged theory. One will see that the algorithm requires essentially only the relative tensor product and relative center of bimodules over fusion $n$-categories. In particular we give the explicit data for ungauging, the inverse to gauging.

\section{Categorical algorithm for gauging}
In this section we discuss the algorithm for gauging. Our starting point is a symmetry assignment 
\[ \phi:\cT\to\cB.\]
Since we focus on fusion $n$-categories in this section, we drop the subscripts for their $E_1$-centers for simplicity, $Z(\cB):=Z_1(\cB)$.
    If one thinks fusion $n$-categories and their bimodules as the data describing potentially anomalous theories and defects between them, the physical meaning of $Z$ is just computing the anomaly, or computing the bulk.  The charge category $Z(\phi)=Z_{\cT|\cB}({}_\phi\cB)= \Fun_{\cT|\cB}({}_\phi\cB,{}_\phi\cB)\cong \Fun_{\cT|\cT}(\cT,{}_\phi\cB_\phi)$ admits a more explicit description, known as the relative center $Z_\cT(\cB)$ (see Appendix~\ref{sec.rc} for more details when $n=1$), which is, informally speaking, the centralizer or ``commutant" of $\phi(\cT)$ in $\cB$ (recall Remark~\ref{rem.charge}). Using the picture that center computes the bulk, we draw these data as Figure~\ref{fig.center}. This graphical representation is essential in our algorithm.
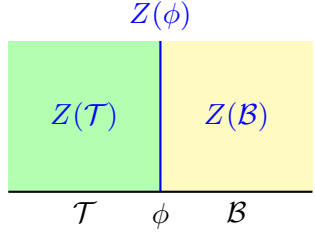
\begin{figure}[ht]
  \centering
\begin{tikzpicture}
  \fill[yellow!30!white] (0,0) rectangle (4,2);
  \fill[green!30!white] (0,0) rectangle (2,2);
  \draw[thick,blue] (2,0) --(2,2) node[above]{$Z(\phi)$};
  \draw[thick] (4,0)--node[below]{$\cB$} (2,0)node[below]{$\phi$} --node[below]{$\cT$} (0,0);
  \node[blue] at (3,1) {$Z(\cB)$};
  \node[blue] at (1,1) {$Z(\cT)$};
\end{tikzpicture}
\caption{Center depicted as bulk.}
\label{fig.center}
\end{figure}

The following two lemmas are useful for later calculation (see also Corollary 3.26 and Remark 3.27 in Ref.~\cite{KZ2107.03858}).
\begin{lemma}
  [Zipping]\label{lem.zip}
  $ \cB\cong \cT\btr[\ov{Z(\cT)}]Z(\phi) $ and under this equivalence the natural monoidal functor 
  \begin{align*}
    \eta_{\cT,Z(\phi)}: \cT &\to \cT\btr[\ov{Z(\cT)}]Z(\phi)\\
    X &\mapsto X \btr[\ov{Z(\cT)}] \one_{Z(\phi)}
  \end{align*}
  coincides with $\phi$.
\end{lemma}
\begin{proof}
  Using the functoriality (Theorem~\ref{thm.fun}) and $\cT\cong \Fun_{\ve[n]|\cT}(\cT,\cT)=Z_{\ve[n]|\cT}(\cT)$:
  \begin{align*}
    \cT\btr[\ov{Z(\cT)}]Z(\phi)&\cong Z_{\ve[n]|\cT}(\cT)\btr[\ov{Z(\cT)}]Z_{\cT|\cB}(_\phi\cB)
    \\&\cong Z_{\ve[n]|\cB}(\cT\btr[\cT]{}_\phi\cB)
    \\&\cong \Fun_{\ve[n]|\cB}(\cB,\cB)\simeq \cB.
  \end{align*}
\GY{
To see how $\eta_{\cT,Z(\phi)}$ coincides with $\phi$, we begin with $s\btr[\ov{Z(\cT)}] F \in\cT\btr[\ov{Z(\cT)}]Z(\phi)$ and its image under the above equivalence is
\begin{align*}
   s\btr[\ov{Z(\cT)}] F \mapsto (s\ot -) \btr[\ov{Z(\cT)}] F \mapsto (s\ot -) \btr[\cT] F \mapsto \phi(s)\ot F(-)\mapsto  \phi(s)\ot F( \one_\cB)\in \cB,
\end{align*}
where we have used the equivalence $\cT\cong \Fun_{\ve[n]|\cT}(\cT,\cT)$, 
\begin{equation*}
    \begin{split}
        \cT&\cong \Fun_{\ve[n]|\cT}(\cT,\cT)\\
        t&\mapsto t\otimes -
    \end{split}
\end{equation*}
 the explicit functoriality map of the center functor given in \cite{KZ1507.00503}, and the equivalence 
 \begin{equation*}
     \begin{split}
         \cT\btr[\cT]{}_\phi \cB&\cong \cB\\
         t\btr[\cB]X&\mapsto \phi(t)\ot X.
     \end{split}
 \end{equation*}
So \begin{equation*}
\begin{split}
    \cT\xrightarrow{\eta_{\cT,Z(\phi)}} &\cT\btr[\ov{Z(\cT)}]Z(\phi) \cong \cB \\
    s\mapsto &s\btr[\ov{Z(\cT)}] \one_{Z(\phi)}\mapsto \phi(s)
\end{split}
\end{equation*} 
 coincides with $\phi$.}
\end{proof}

\begin{lemma}
  [Unzipping]\label{lem.unzip}
  Suppose that $\cC$ is a right multi-fusion $Z(\cT)$-module and $\cB\cong \cT\btr[\ov{Z(\cT)}]\cC$. Let $\phi$ be the composition $ \cT\xrightarrow{\eta_{\cT,\cC}} \cT\btr[\ov{Z(\cT)}]\cC \cong \cB$, we have $\cC\cong Z(\phi).$
\end{lemma}
\begin{proof}
  \begin{align*}
    \cC&= Z(\cT)\btr[\ov{Z(\cT)}] \cC=\Fun_{\cT|\cT}(\cT,\cT)\btr[\ov{Z(\cT)}] \cC
    \\&=\Fun_{\cT|\cT}(\cT,\cT\bt\cT^\rev)\btr[\cT\bt\cT^\rev] \cT\btr[\ov{Z(\cT)}] \cC\\ &=\Fun_{\cT|\cT}(\cT,\cT\btr[\ov{Z(\cT)}] \cC)\cong \Fun_{\cT|\cT}(\cT,{}_\phi\cB_\phi)\cong Z(\phi).
  \end{align*}
\end{proof}
Graphically, zipping and unzipping can be presented as Figure~\ref{fig.zip}.
\begin{figure}[ht]
  \centering
\begin{tikzpicture}
  \fill[yellow!30!white] (0,0) rectangle (4,2);
  \fill[green!30!white] (0,0)--(2,0) --(2,0.5) arc [start angle=0, end angle=90, radius=0.5]
  --(0,1) -- cycle;
  \draw[thick,blue] (2,0) --(2,0.5)  arc  [start angle=0, end angle=90, radius=0.5] 
  --node[above right]{$Z(\phi)$} (0,1);
  \draw[thick] (4,0)--node[below]{$\cB$} (2,0)node[below]{$\phi$} --node[below]{$\cT$} (0,0);
  \node[blue] at (3,1) {$Z(\cB)$};
  \node[blue] at (1,0.5) {$Z(\cT)$};
  \node[blue,above] at (2,2) {$\cT\btr[\ov{Z(\cT)}]Z(\phi)\cong \cB$};
\end{tikzpicture}\ \ \ \ 
\begin{tikzpicture}
  \fill[yellow!30!white] (0,0) rectangle (4,2);
  \fill[green!30!white] (0,0)--(4,0) -- (4,0.2) -- (2.5,0.2) arc [start angle=-90, end angle=-180, radius=0.5] --(2,2) --(0,2) -- cycle;
  \draw[thick,blue] (4,0.2) --node[above]{$\cC$} (2.5,0.2) arc [start angle=-90, end angle=-180, radius=0.5] --(2,2)node[above]{$\cC\cong Z(\eta_{\cT,\cC}:\cT\to \cT\btr[\ov{Z(\cT)}]\cC)$};
  \draw[thick] (4,0)--node[below]{$\cT$} (0,0);
  \node[blue] at (3,1.3) {$Z(\cT\btr[\ov{Z(\cT)}]\cC)$};
  \node[blue] at (1,1) {$Z(\cT)$};
\end{tikzpicture}
\caption{Zipping and unzipping.}
\label{fig.zip}
\end{figure}
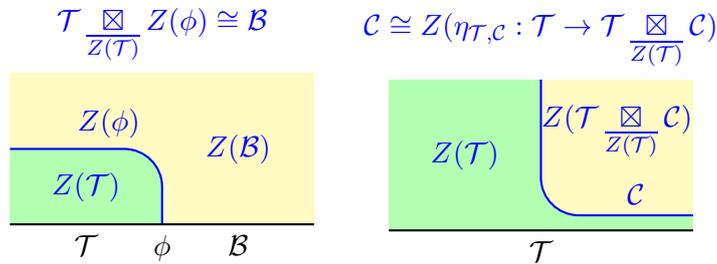

\begin{remark}
  The notion similar to $(\cC,\cT\btr[\ov{Z(\cT)}]\cC\cong \cB)$ first appeared in~\cite{KWZ1502.01690} as a \emph{morphism} between topological orders. Later in~\cite{FMT2209.07471} it is independently proposed again as the \emph{topological symmetry} of $\cB$, often also referred to as ``sandwich''. Here we further clarify the physical meanings of these data. According to the zipping and unzipping lemmas (see also~\cite{KZ1507.00503}), the data of monoidal functor $\phi:\cT\to \cB$ is equivalent to the data of the pair $(\cC,\cT\btr[\ov{Z(\cT)}]\cC\cong \cB)$. Here $\phi$ is the symmetry assignment of $\cB$, and also an object $\Sigma\phi$ in the category of SET orders $\Fun(\Sigma\cT,\Sigma\cB)$. As we have discussed, $\phi$ automatically determines 
  \begin{itemize} \item the charge or representation category $Z(\phi)$ consisting of operators/defects invariant while passing through the image of $\phi$, i.e., invariant under the symmetry action.
    \item the background category~\cite{KZ1705.01087,CJK+1903.12334,KZ1905.04924,KZ1912.01760,LY2208.01572}, categorical symmetry~\cite{JW1905.13279,JW1912.13492,KLW+2005.14178,JW2106.02069}, SymTFT~\cite{ABE+2112.02092,BS2304.02660,BS2305.17159}, symmetry TO~\cite{CW2203.03596,CW2205.06244} or quantum currents~\cite{LZ2305.12917}, which serves as the higher dimensional ``glue'',
    \item the \emph{quiche}~\cite{FMT2209.07471} $(\cT,Z(\cT))$,
    \item and the equivalence $\cT\btr[\ov{Z(\cT)}] Z(\phi)\cong \cB$.
  \end{itemize}
  Conversely $\cT\btr[\ov{Z(\cT)}]\cC\cong \cB$ automatically implies the symmetry assignment, i.e., monoidal functor $\cT\xrightarrow{\eta_{\cT,\cC}} \cT\btr[\ov{Z(\cT)}]\cC \cong \cB$, such that the charge category is exactly $\cC$.
\end{remark}

\begin{definition}
    An abstract gauging of the abstract symmetry $\cT$ is an indecomposable right $\cT$-module $\cK$.  The fusion $n$-category dual to $\cT$ with respect to $\cK$, $\cT_\cK^\vee:=Z_{\ve[n]|\cT}(\cK)=\Fun_{\ve[n]|\cT}(\cK,\cK)$, is physically referred to as the gauge symmetry or dual symmetry (of $\cT$ with respect to $\cK$).
\end{definition}
\begin{remark}
  Equivalently, an abstract gauging can be specified by an (Morita class of) algebra $A$ in $\cT$ such that $\lmd A \cT\cong \cK$. The symmetry defects contained in the algebra $A$ are to be ``summed over'' or ``condensed''. It is not hard to check that $\cT_\cK^\vee=\bmd A A \cT$.
    $\cT_\cK^\vee$ is Morita equivalent to $\cT$ by definition (see Remark~\ref{rem.mor}).
\end{remark}

\begin{remark} \label{rem.gtype}
\GY{We 
can categorize abstract gauging into the following types
\begin{enumerate}
    \item Complete gauging. If the abstract gauging $\cK$ is equivalent to $\ve[n]$ as an $(n+1)$-vector space (forgetting the $\cT$-module structure), $\cK$ is a complete gauging. In this case, physically, all the symmetry defects are condensed. We also know that there exists at least one local structure for $\cT$, given by the $\cT$-module structure on $\cK$, $\cT\to \Fun(\cK,\cK)\cong \ve[n]$. In practice, we usually fix a local structure $F:\cT\rightarrow \ve[n]$, and complete gauging can be further classified as
    \begin{enumerate}
        \item Ordinary gauging: $\cK=\ve[n]_F$, where the $\cT$-module structure is induced by the local structure $F$. See Examples~\ref{eg.OGT}, \ref{eg.DWGT} and \ref{eg.2set}. 
        \item Twisted gauging.  As we have mentioned, given a fusion $n$-category $\cT$, the monoidal functor $\cT\to \ve[n]$, if exists, is usually not unique. By a twisted gauging, we mean $\cK'=\ve[n]_{F'}$, whose $\cT$-module structure is induced by a monoidal functor $F':\cT\to\ve[n]$ other than the local structure $F$. See Example~\ref{eg.TG}. On the other hand, the ordinary gauging may be called the untwisted gauging.
    \end{enumerate}
    \item Partial gauging: taking the abstract gauging $\cK\neq\ve[n]$, or physically, condensing only part of the symmetry defects.
\end{enumerate}
}
\begin{example}
    For a local $0$-form symmetry $\cT=\ve[n]_G$, upon ordinary gauging, the gauge symmetry is $(\ve[n]_G)_{\ve[n]}^\vee\cong \Rep[n] G\cong\Sigma^{n-1}\Rep G$, an $(n-1)$-symmetry, generated by operators acting on spatial points.
\end{example}

\begin{definition}[Gauging {[Morita-invariant version]}]
\label{def.gauging}
     Given a bare theory $\cB$, we need to specify its symmetry $\phi:\cT\to \cB$, and an abstract gauging $\cK$ of $\cT$, and then the gauged theory with respect to these choices is defined to be $\cB^\phi_\cK:=\cT_\cK^\vee \btr[\ov{Z(\cT)}] Z(\phi) \cong Z_{\ve[n]|\cB}(\cK\btr[\cT] {}_\phi  \cB)$.
\end{definition}

\begin{theorem}[Gauging is reversible]
\label{thm.gauging}
     The gauged theory $\cB^\phi_\cK$ is naturally equipped with the gauge symmetry $\cT_\cK^\vee$ via the monoidal functor
    $\eta_{\cT_\cK^\vee,Z(\phi)}:\cT_\cK^\vee\to \cT_\cK^\vee \btr[\ov{Z(\cT)}] Z(\phi)=\cB^\phi_\cK$,
    and moreover $\cK^\vee$ is naturally a right $\cT_\cK^\vee$-module,
    with the dual $(\cT_\cK^\vee)_{\cK^\vee}^\vee\cong\cT$. 
    Gauging $\cB^\phi_\cK$ with respect to these choices gives the original theory,
     $(\cB^\phi_\cK)^{\eta_{\cT_\cK^\vee,Z(\phi)}}_{\cK^\vee}\cong \cB$.
\end{theorem}
\begin{proof}
    By unzipping (Lemma~\ref{lem.unzip}) we know $Z({\eta_{\cT_\cK^\vee,Z(\phi)}})\cong Z(\phi)$. Therefore, by zipping (Lemma~\ref{lem.zip})
    \begin{align*}
      (\cB^\phi_\cK)^{\eta_{\cT_\cK^\vee,Z(\phi)}}_{\cK^\vee}&\cong(\cT_\cK^\vee)_{\cK^\vee}^\vee\btr[\ov{Z(\cT_\cK^\vee)}]Z({\eta_{\cT_\cK^\vee,Z(\phi)}})
      \\&\cong \cT\btr[\ov{Z(\cT )}]Z(\phi)\cong \cB.
    \end{align*}
\end{proof}
\begin{remark}
  Again, suppose that $\cK=\lmd A\cT$, one can check that $\cB_\cK^\phi=Z_{\ve[n]|\cB}(\lmd A \cT\btr[\cT]{}_\phi\cB)=Z_{\ve[n]|\cB}(\lmd{\phi(A)}\cB)=\bmd{\phi(A)}{\phi(A)}\cB$, $\eta_{\cT_\cK^\vee,Z(\phi)}=\Omega_A\Sigma\phi$, and $\cK^\vee=\rmd A\cT$. Therefore, our two defintions for gauging, Definition~\ref{def.gaugingalg} in terms of algebras in $\cT$ and Definition~\ref{def.gauging} in terms of $\cT$-modules, are equivalent.
\end{remark}
\begin{center}
    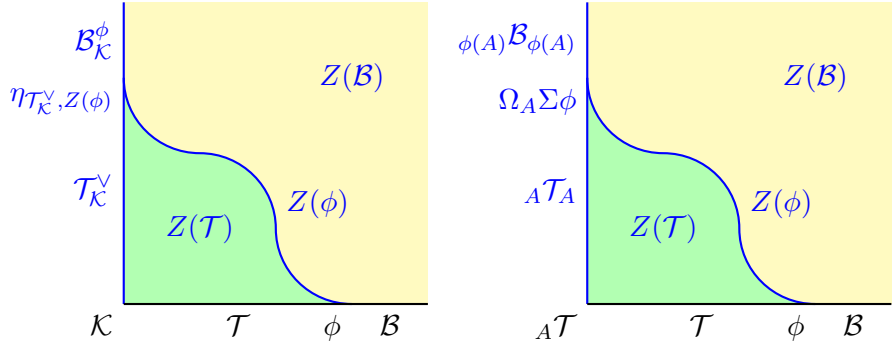
\begin{figure}[ht]
   \centering
\begin{tikzpicture}[baseline=(current bounding box.center),scale=1]
  \fill[yellow!30!white] (0,0) rectangle (4,4);
  \fill[green!30!white] (0,0)--(3,0) arc [start angle=-90, end angle=-180, radius=1] arc [start angle=0, end angle=90, radius=1]
  arc [start angle=-90, end angle=-180, radius=1] -- cycle;
  \draw[thick,blue] (3,0) arc [start angle=-90, end angle=-180, radius=1] node[above right]{$Z(\phi)$} arc  [start angle=0, end angle=90, radius=1] 
  arc [start angle=-90, end angle=-180, radius=1];
  \draw[thick,blue] (0,0)--node[left]{$\cT_\cK^\vee$} (0,3) node[below left]{$\eta_{\cT_\cK^\vee,Z(\phi)}$} --node[left]{$\cB_\cK^\phi$} (0,4);
  \draw[thick] (4,0)--node[below]{$\cB$} (3,0)node[below left]{$\phi$} --node[below]{$\cT$} (0,0)node[below left]{$\cK$};
  \node[blue] at (3,3) {$Z(\cB)$};
  \node[blue] at (1,1) {$Z(\cT)$};
\end{tikzpicture}~
\begin{tikzpicture}[baseline=(current bounding box.center),scale=1]
  \fill[yellow!30!white] (0,0) rectangle (4,4);
  \fill[green!30!white] (0,0)--(3,0) arc [start angle=-90, end angle=-180, radius=1] arc [start angle=0, end angle=90, radius=1]
  arc [start angle=-90, end angle=-180, radius=1] -- cycle;
  \draw[thick,blue] (3,0) arc [start angle=-90, end angle=-180, radius=1] node[above right]{$Z(\phi)$} arc  [start angle=0, end angle=90, radius=1] 
  arc [start angle=-90, end angle=-180, radius=1];
  \draw[thick,blue] (0,0)--node[left]{${}_A\cT_A$} (0,3) node[below left]{$\Omega_A \Sigma\phi$} --node[left]{${}_{\phi(A)}\cB_{\phi(A)}$} (0,4);
  \draw[thick] (4,0)--node[below]{$\cB$} (3,0)node[below left]{$\phi$} --node[below]{$\cT$} (0,0)node[below left]{$_A\cT$};
  \node[blue] at (3,3) {$Z(\cB)$};
  \node[blue] at (1,1) {$Z(\cT)$};
\end{tikzpicture}
\caption{Graphical representation for generalized gauging, in terms of $\cT$-module $\cK$ v.s. in terms of algebra $A\in \cT$ where $\cK=\lmd A\cT.$}
\label{fig.gauge}
\end{figure}

\begin{figure}[ht]

   \centering
  \begin{tikzpicture}[baseline=(current bounding box.center)]
  \fill[yellow!30!white] (0,0) rectangle (4,4);
  \fill[green!30!white] (0,0)--(1,0) arc [start angle=-90, end angle=-180, radius=1]  -- cycle;
  \fill[green!30!white] (0,4)--(1,4) arc [start angle=90, end angle=180, radius=1]  -- cycle;
  \draw[thick,blue] (1,0) arc [start angle=-90, end angle=-180, radius=1] ;
  \draw[thick,blue] (1,4) arc [start angle=90, end angle=180, radius=1] ;
  \draw[thick,blue] (0,0)--node[left]{$\cT_\cK^\vee$} (0,1)node[left]{$\eta_{\cT_\cK^\vee,Z(\phi)}$} --node[left]{$\cB_\cK^\phi$} (0,3)node[left]{$\eta_{\cT_\cK^\vee,Z(\phi)}$}
  --node[left]{$\cT_\cK^\vee$} (0,4) ;
  \draw[thick] (4,0)--node[below]{$\cB$} (1,0)node[below right]{$\phi$} --node[below]{$\cT$} (0,0)node[below left]{$\cK$};
  \draw[thick] (4,4)--node[above]{$\cB$} (1,4)node[above right]{$\phi$} --node[above]{$\cT$} (0,4)node[above left]{$\cK^\vee$};
  \node[blue] at (2,2) {$Z(\cB)$};
  \node[blue] at (.6,.6) {$Z(\phi)$};
  \node[blue] at (1.2,3.4) {$Z(\eta_{\cT_\cK^\vee,Z(\phi)})$};
  \end{tikzpicture}
  \ \ =\ \ 
  \begin{tikzpicture}[baseline=(current bounding box.center)]
  \fill[yellow!30!white] (0,0) rectangle (4,4);
  \fill[green!30!white] (0,0)--(3,0) arc [start angle=-90, end angle=-180, radius=1] -- (2,3)
  arc [start angle=180, end angle=90, radius=1] -- (0,4) -- cycle;
  \draw[thick,blue] (3,0) arc [start angle=-90, end angle=-180, radius=1] node[below left]{$Z(\phi)$} --  (2,3) arc  [start angle=180, end angle=90, radius=1];
  \draw[thick,blue] (0,0)--node[left]{$\cT_\cK^\vee$} (0,4);
  \draw[thick] (4,0)--node[below]{$\cB$} (3,0)node[below left]{$\phi$} --node[below]{$\cT$} (0,0)node[below left]{$\cK$};
  \draw[thick] (4,4)--node[above]{$\cB$} (3,4)node[above left]{$\phi$} --node[above]{$\cT$} (0,4)node[above left]{$\cK^\vee$};
  \node[blue] at (3,2) {$Z(\cB)$};
  \node[blue] at (1,2) {$Z(\cT)$};
  \end{tikzpicture}
  \\
  \bigskip

  \begin{tikzpicture}[baseline=(current bounding box.center)]
  \fill[yellow!30!white] (0,0) rectangle (4,4);
  \fill[green!30!white] (0,0)--(1,0) arc [start angle=-90, end angle=-180, radius=1]  -- cycle;
  \fill[green!30!white] (0,4)--(1,4) arc [start angle=90, end angle=180, radius=1]  -- cycle;
  \draw[thick,blue] (1,0) arc [start angle=-90, end angle=-180, radius=1] ;
  \draw[thick,blue] (1,4) arc [start angle=90, end angle=180, radius=1] ;
  \draw[thick,blue] (0,0)--node[left]{$_A\cT_A$} (0,1)node[left]{$\Omega_A \Sigma\phi$} --node[left]{${}_{\phi(A)}\cB_{\phi(A)}$} (0,3)node[left]{$\Omega_A \Sigma\phi$}
  --node[left]{${}_A\cT_A$} (0,4) ;
  \draw[thick] (4,0)--node[below]{$\cB$} (1,0)node[below right]{$\phi$} --node[below]{$\cT$} (0,0)node[below left]{$_A\cT$};
  \draw[thick] (4,4)--node[above]{$\cB$} (1,4)node[above right]{$\phi$} --node[above]{$\cT$} (0,4)node[above left]{$\cT_A$};
  \node[blue] at (2,2) {$Z(\cB)$};
  \node[blue] at (1.8,.6) {$Z(\phi)=Z(\Omega_\bullet\Sigma\phi)$};
  \node[blue] at (1.2,3.4) {$Z(\Omega_A \Sigma\phi)$};
  \end{tikzpicture}
  \ \ =\ \ 
  \begin{tikzpicture}[baseline=(current bounding box.center)]
  \fill[yellow!30!white] (0,0) rectangle (4,4);
  \fill[green!30!white] (0,0)--(3,0) arc [start angle=-90, end angle=-180, radius=1] -- (2,3)
  arc [start angle=180, end angle=90, radius=1] -- (0,4) -- cycle;
  \draw[thick,blue] (3,0) arc [start angle=-90, end angle=-180, radius=1] node[below left]{$Z(\phi)$} --  (2,3) arc  [start angle=180, end angle=90, radius=1];
  \draw[thick,blue] (0,0)--node[left]{${}_A\cT_A$} (0,4);
  \draw[thick] (4,0)--node[below]{$\cB$} (3,0)node[below left]{$\phi$} --node[below]{$\cT$} (0,0)node[below left]{$_A\cT$};
  \draw[thick] (4,4)--node[above]{$\cB$} (3,4)node[above left]{$\phi$} --node[above]{$\cT$} (0,4)node[above left]{$\cT_A$};
  \node[blue] at (3,2) {$Z(\cB)$};
  \node[blue] at (1,2) {$Z(\cT)$};
  \end{tikzpicture}
  \caption{Gauging is reversible.}
  \label{fig.gr}
\end{figure}
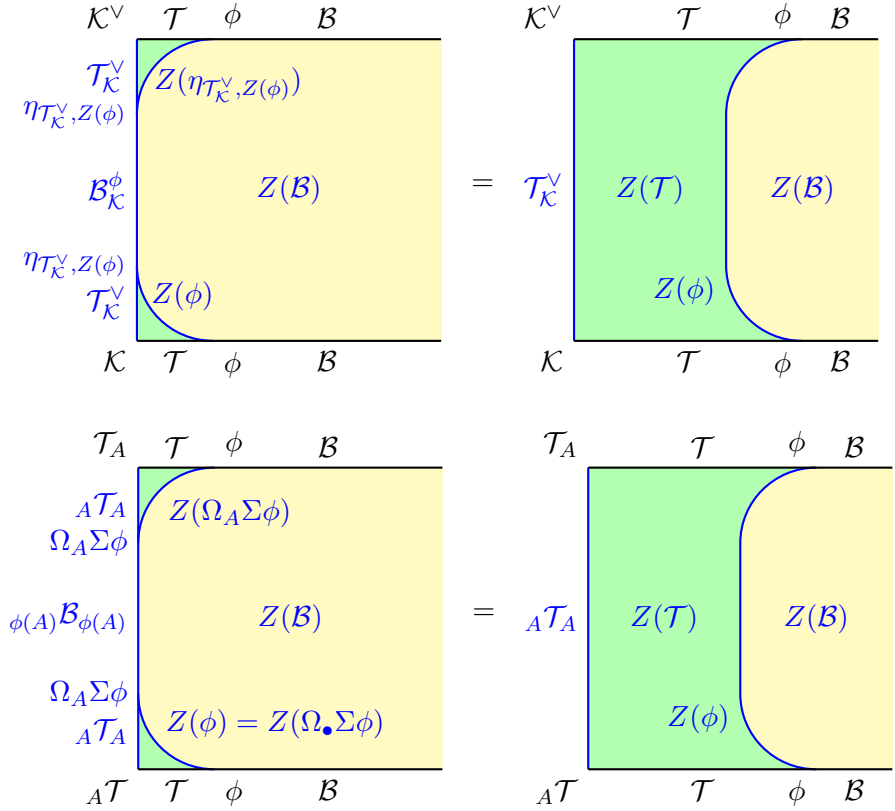
\end{center}    
The gauging procedure can be depicted as Figure~\ref{fig.gauge} where the zipping and unzipping is also drawn in a more intuitive way. Note that $T_\cK^\vee$ is by definition the center of the right $\cT$-module $\cK$, in Figure~\ref{fig.gauge} the bottom line is the macroscopic/categorical data needed for gauging, and the gauging procedure is simply computing the (relative) center, or graphically the bulk, of all these data, resulting in a ``corner'' graph. The reversibility of gauging is concluded in Figure~\ref{fig.gr}.

  It can be seen directly from Definition~\ref{def.gauging} and Theorem~\ref{thm.gauging} that $\cK\btr[\cT]{}_\phi\cB$ is an invertible $\cB_\cK^\phi$-$\cB$-bimodule, and gauging is a special type of Morita equivalence, which preserves the center (physically the anomaly, the bulk, or the quantum currents) $Z(\cB)$ (see Remark~\ref{rem.mor}). 
\begin{remark}
    If we assign $\cB$ itself as the symmetry of $\cB$, i.e., choosing $\phi=\id_\cB:\cB\to \cB$, $Z(\id_\cB) = \Fun_{\cB|\cB}(\cB,\cB)= Z(\cB)$, then gauging exhausts all the Morita equivalences. Each gauging
    $\cB_\cK^\id = \cB_{\cK}^\vee\btr[\ov{Z(\cB)}]Z(\cB)\cong \cB_{\cK}^\vee$ , corresponds to an object $\cK$ in $\rmd{\cB}{ \ve[n+1] }\cong \Sigma\cB$.
\end{remark}

\end{remark}
\begin{remark}
    The ordinary gauging $\cK=\ve[n]$ in such categorical/holographic picture has been studied in previous works, under the name \emph{categorical gauging}~\cite{KLW+2005.14178} or \emph{quotient}~\cite{FMT2209.07471}. 
\end{remark}
\GY{
\begin{remark}\label{rem.gsub}
     A subsymmetry $\cS$ of a fusion $n$-category symmetry $\cT$  is a fusion $n$-category with a monoidal embedding $\eta:\cS\rightarrow \cT$. Note that partially gauging a symmetry does not always mean a complete gauging of a subsymmetry. Let's take $\cK=\lmd A{\cT}\neq \ve[n]$ as an abstract gauging of $\cT$. It  can be realized as a complete gauging of the subsymmetry $\cS$ if there is an algebra $B\in\cS$, such that $\ve[n]=\lmd B{\cS}$, and  $ \cK = \ve[n]\btr[\cS]{}_\eta \cT$ (see Figure~\ref{fig.pgsubg}), or 
        $$\lmd A{\cT} = \lmd B{\cS}\btr[S]{}_\eta \cT=\lmd {\eta(B)} {\cT}.$$ 
In other words, $A$ and $\eta(B)$ are Morita equivalent. If such algebra $B$ does not exist, the partial gauging $\cK=\lmd A{\cT}$ can not be realized by completely gauging the subsymmetry $\cS$.  See Examples~\ref{eg.subsym} and \ref{eg.pgnonsub}.
\end{remark}}\begin{figure}[ht]
    \centering
    \begin{tikzpicture}[baseline=(current bounding box.center),scale=0.8]
    \fill[green!30!white] (0,0)--(6,0) arc [start angle=-90, end angle=-180, radius=2] arc [start angle=0, end angle=90, radius=2]
    arc [start angle=-90, end angle=-180, radius=2] -- cycle;
    \draw[thick,blue] (6,0) arc [start angle=-90, end angle=-180, radius=2] node[above right]{$Z(\phi)$} arc  [start angle=0, end angle=90, radius=2] 
    arc [start angle=-90, end angle=-180, radius=2];
    \fill[cyan!30!white] (0,0)--(3,0) arc [start angle=-90, end angle=-180, radius=1] arc [start angle=0, end angle=90, radius=1]
    arc [start angle=-90, end angle=-180, radius=1] -- cycle;
    \draw[thick,blue] (3,0) arc [start angle=-90, end angle=-180, radius=1] node[above right]{$Z(\eta)$} arc  [start angle=0, end angle=90, radius=1] 
    arc [start angle=-90, end angle=-180, radius=1];
    \draw[thick,blue] (0,0) -- node[left]{$\cS_{\ve[n]}^\vee$} (0,3) -- node[left]{$\cT_\cK^\vee$} (0,4.5) node[below left]{$ $} --node[left]{$\cB_\cK^\phi$} (0,8);
    \draw[thick] (8,0)--node[below]{$\cB$} (6,0)
    node[below left]{$\phi$}
    --node[below]{$\cT$} (3,0) node[below left]{$\eta$}
    --node[below]{$\cS$} (0,0)node[below left]{$\ve[n]$};
    \node[blue] at (6,6) {$ $};
    \node[blue] at (2.5,2.5) {$Z(\cT)$};
    \node[blue] at (1,1) {$Z(\cS)$};
    \draw [ thick, decoration={
        brace,
        mirror,
    },
    decorate
] (-1,-0.7) --node [below] {$\cK$} (3,-0.7)  ; 
  \end{tikzpicture}
    \caption{Gauging a subsymmetry as partial gauging}
    \label{fig.pgsubg}
   \end{figure}
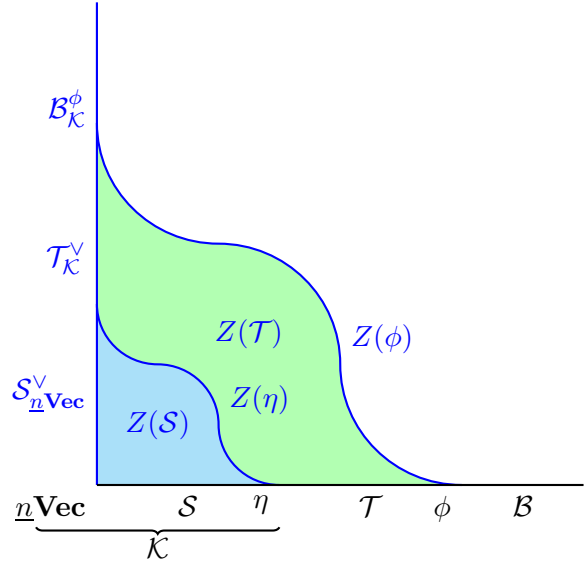

  When $Z(\cB)=\ve[n]$, the original theory $\cB$ and the gauged theory $\cB_\cK^\phi$ are both anomaly-free theories.
  When $Z(\cB)$ is nontrivial, we performed a gauging on the $n$+1D boundary while preserving the bulk.
  As we allow a nontrivial bulk, the higher gauging~\cite{RSS2204.02407} is automatically covered by our framework.
  It is also interesting to consider the ``lower'' gauging, i.e., for a boundary theory $\cB$ we try to gauge its bulk $Z(\cB)$. Based on previous discussions, we should first deloop $Z(\cB)$ to the fusion $(n+1)$-category $\Sigma Z(\cB)=\Fun_{\ve[n+1]|\ve[n+1]}(\Sigma\cB,\Sigma\cB)$, and then specify the symmetry of $\Sigma Z(\cB)$ to be gauged. We are in particular interested in the case that after gauging the bulk becomes trivial, allowing us to obtain an anomaly-free theory still in $n$+1D.
  Let $\cT$ be a fusion $(n+1)$-category with a symmetry assignment $\phi:\cT\to \Sigma Z(\cB)$ and an abstract gauging (right indecomposible $\cT$-module) $\cK$. The condition that the bulk $\Sigma Z(\cB)$ becomes trivial after gauging:  $\left(\Sigma Z(\cB)\right)_{\cK}^{\phi}=Z_{\ve[n+1]|\Sigma Z(\cB)}(\cK\btr[\cT] {}_{\phi}\Sigma Z(\cB))=\ve[n+1]$, is equivalent to $Z_{\ve[n+1]|\ve[n+1]}(\cK\btr[\cT] {}_{\phi}\Sigma Z(\cB))=(\Sigma Z(\cB))^\rev$, i.e., $\cK\btr[\cT] {}_{\phi}\Sigma Z(\cB)$ is a collection of boundary conditions of $Z(\cB)$. If we care only about the $n$+1D theory, the data needed for bulk-gauging can be simplified:

\begin{definition}
  [Bulk-gauging]  Take a (multi-)fusion $n$-category $\cF$ and a braided equivalence $\beta: \ov{Z(\cF)}\cong Z(\cB)$, i.e. $\cF$ is another boundary of $Z(\cB)$ which is not necessarily Morita equivalent to $\cB$. We define the bulk-gauged theory to be the ``sandwich'' $$\cF_\beta\btr[{Z(\cB)}]\cB.$$
\end{definition}

\begin{example}
  [Gauging fermion parity in 2+1D (16-fold way)] The fusion 2-category describing a 2+1D fermionic system is $\sve[2]$ which has a nontrivial center $Z(\sve[2])$. Gauging fermion parity is in fact a bulk-gauging, where we choose the bulk symmetry to be still $\cT=\Sigma Z(\sve[2])$ and the fermion parity fluxes in $Z(\sve[2])$ are to be condensed, i.e., $\cK=\Sigma \sve[2]$ and $\cT_{\cK}^\vee=\ve[3]$. There are 16 choices for braided equivalence $\beta: Z(\sve[2])\to Z(\sve[2])$, and the bulk-gauged theories are $\Sigma(\cM_\beta)$ where $\cM_\beta$ are the 16 minimal modular extensions of $\sve$. The ungauged fermionic invertible phases should really be described by the enriched fusion 2-category $^{Z(\sve[2])}{}_\beta \sve[2]$ where the action of $Z(\sve[2])$ on $\sve[2]$ is the forgetful functor composed with $\beta$.
\end{example}

\section{Examples of gauging}
\begin{example}[Ordinary gauge theory]
\label{eg.OGT}
  Pick a group $G$, and choose $\cA=\ve[n]$ to be the trivial phase, 
  $\cT=\ve[n]_G$ to be the ordinary global symmetry, $\phi:\ve[n]_G\to \ve[n]$ to be functor that forgets the $G$-grading, and $\cK=\ve[n]$ the ordinary gauging. We have $Z(\phi)=\Rep[n] G$ and $(\ve[n]_G)_{\ve[n]}^\vee=\Rep[n] G$.
   The gauged theory is $\ve[n]_{\ve[n]}^\phi=\Rep[n] G\btr[\ov{Z(\ve[n]_G )}]\Rep[n] G=\Fun_{\ve[n]|\ve[n]}(\Rep[n] G,\Rep[n] G)$. For $n\geq 2$, $\ve[n-1]_G$ is an algebra in $\ve[n]_G$, and $\lmd{\ve[n-1]_G}{ \ve[n]_G }= \ve[n]$, $\rmd{\ve[n-1]_G}{ \ve[n]_G }= \ve[n]$, so  
   \begin{align*}
      \ve[n]\btr[{\ve[n]_G}]  \ve[n] &\cong \bmd{\ve[n-1]_G}{\ve[n-1]_G}{ \ve[n]_G }
      \\&\cong \rmd{\ve[n-1]_G}{ \ve[n] }\cong \Sigma\ve[n-1]_G,
   \end{align*}
   then the gauged theory is
  \begin{equation*}
  \begin{split}
      &\Rep[n] G\btr[\ov{Z(\ve[n]_G)}] \Rep[n] G\\
      & = Z_{\ve[n]|\ve[n]_G}(\ve[n])\btr[\ov{Z(\ve[n]_G)}] 
      Z_{\ve[n]_G|\ve[n]}(\ve[n])\\
      &\cong Z_{\ve[n]|\ve[n]}( \ve[n]\btr[{\ve[n]_G}] \ve[n]) = Z_0( \Sigma (\ve[n-1]_G))\\
      & = \Sigma Z(\ve[n-1]_G),
  \end{split}
  \end{equation*}
   which is the $G$-gauge theory in $n$+1D. For the second equality, we have used the functoriality of the center functor. (For $n=1$ the gauged theory is a multi-fusion category, which by negative thinking can also be thought as the $G$-gauge theory in 1+1D, though it is not stable.)
   \begin{figure}[ht]
    \centering
\begin{tikzpicture}[baseline=(current bounding box.center)]
    \fill[green!30!white] (0,0)--(3,0) arc [start angle=-90, end angle=-180, radius=1] arc [start angle=0, end angle=90, radius=1]
    arc [start angle=-90, end angle=-180, radius=1] -- cycle;
    \draw[thick,blue] (3,0) arc [start angle=-90, end angle=-180, radius=1] node[above right]{$\Rep[n] G$} arc  [start angle=0, end angle=90, radius=1] 
    arc [start angle=-90, end angle=-180, radius=1];
    \draw[thick,blue] (0,0)--node[left]{$\Rep[n] G$} (0,3) node[below left]{$ $} --node[right]{$Z_{\ve[n]|\ve[n]}(\Rep[n] G)$} (0,4);
    \draw[thick] (4,0)--node[below]{$\ve[n]$} (3,0)node[below left]{$\phi$} --node[below]{$\ve[n]_G$} (0,0)node[below left]{$\ve[n]$};
    \node[blue] at (3,3) {$ $};
    \node[blue] at (1,1) {$Z(\ve[n]_G)$};
  \end{tikzpicture}
    
    \caption{Ordinary gauge theory}
    \label{fig.ogt}
   \end{figure}
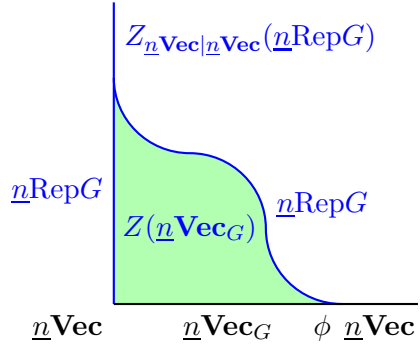
   
\end{example}
\begin{example}
  [Dijkgraaf-Witten gauge theory]
  \label{eg.DWGT}
  Pick the same choice as in the above example except that $\phi^{\omega_{n+1}}:\ve[n]_G\to \ve[n]$ is chosen to be twisted by an $(n+1)$-cocycle $\omega_{n+1}\in H^n(G,U(1))$, and we have identification $_{\phi^{\omega_{n+1}}} \ve[n]\cong \rmd{\ve[n-1]_G^{\omega_{n+1}}}{ \ve[n]_G }$. This choice describes a SPT phase. 
  We have $Z(\phi^{\omega_{n+1}})=\Rep[n] G$ with an $\omega_{n+1}$ twisted action of $Z(\ve[n]_G)$. 
  Similar as the ordinary gauge theory, since
  \begin{align*}
      &\ve[n]\btr[{\ve[n]_G}] {}_{\phi^{\omega{n+1}}} \ve[n] \cong \bmd{\ve[n-1]_G}{\ve[n-1]_G^{\omega_{n+1}}}{ \ve[n]_G }
      \\&\cong \rmd{\ve[n-1]_G^{\omega_{n+1}}}{ \ve[n] }\cong \Sigma\ve[n-1]_G^{\omega_{n+1}},
  \end{align*} then the gauged theory is
  \begin{equation*}
      \begin{split}
          &Z_{\ve[n]|\ve[n]}( \ve[n]\btr[{\ve[n]_G}] {}_{\phi^{\omega{n+1}}} \ve[n])\\& = Z_0( \Sigma (\ve[n-1]_G^{\omega_{n+1}})) = \Sigma Z(\ve[n-1]_G^{\omega_{n+1}}),
      \end{split}
  \end{equation*}
  which is the Dijkgraaf-Witten gauge theory.
  \begin{figure}[ht]
    \centering
\begin{tikzpicture}[baseline=(current bounding box.center)]
    \fill[green!30!white] (0,0)--(3,0) arc [start angle=-90, end angle=-180, radius=1] arc [start angle=0, end angle=90, radius=1]
    arc [start angle=-90, end angle=-180, radius=1] -- cycle;
    \draw[thick,blue] (3,0) arc [start angle=-90, end angle=-180, radius=1] node[above right]{$\Rep[n] G$} arc  [start angle=0, end angle=90, radius=1] 
    arc [start angle=-90, end angle=-180, radius=1];
    \draw[thick,blue] (0,0)--node[left]{$\Rep[n] G$} (0,3) node[below left]{$ $} --node[right]{$Z_{\ve[n]|\ve[n]}(\Sigma\ve[n-1]_G^{\omega_{n+1}})$} (0,4);
    \draw[thick] (4,0)--node[below]{$\ve[n]$} (3,0)node[above]{$\phi^{\omega_{n+1}}$} --node[below]{$\ve[n]_G$} (0,0)node[below left]{$\ve[n]$};
    \node[blue] at (3,3) {$ $};
    \node[blue] at (1,1) {$Z(\ve[n]_G)$};
  \end{tikzpicture}

    \caption{Dijkgraaf-Witten gauge theory}
    \label{fig.DWgt}
  \end{figure}
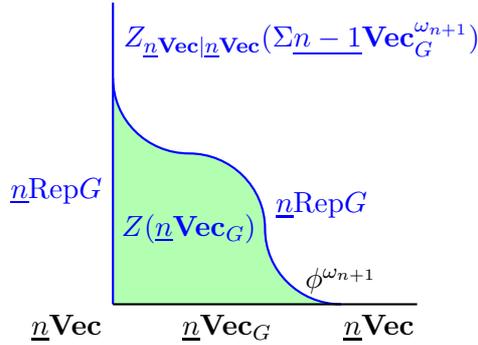
\end{example}
\begin{example}
    [Twisted gauging]
    \label{eg.TG}
    \TL{The twisted fiber functor $\phi^{\omega_{n+1}} :\ve[n]_G\to\ve[n]$ can also be used to equip $\ve[n]$ with a different right $\ve[n]_G$-module structure, denoted by $\ve[n]_{\phi^{\omega_{n+1}}}$. Performing twisted gauging (with $\ve[n]_{\phi^{\omega_{n+1}}}$) to a SPT phase $\phi^{\nu_{n+1}} :\ve[n]_G\to\ve[n]$, the gauged theory is $\Sigma Z(\ve[n-1]_G^{\nu_{n+1}-\omega_{n+1}})$.  
     In 1+1D, such twisted gauging can be realized by the Kennedy-Tasaki transformation~\cite{LOZ2301.07899,SS2404.01369,LSY2405.14939}. 
    }
\end{example}
\begin{example}
  [2+1D Bosonic SET]\label{eg.2set}
  Take a MTC $\cM$ equipped with a $G$-action. By Ref.~\cite{BBCW1410.4540}, the $G$-action can be gauged if and only if there exists a $G$-crossed braided extension $\cM_G^\xt$, and in the 2-categorical language, if and only if there exists a monoidal functor $\phi:\ve[2]_G\to \Sigma\cM$ (see~\cite{ENO0909.3140} Theorem 7.12, noting that the categorical Picard 2-group is a subcategory of $\rmd \cM {\ve[2]}\cong \Sigma\cM$). Here the abstract gauging is ordinary $\cK=\ve[2]$ and gauge symmetry is $\Rep[2] G$. We have $Z(\phi)=\Sigma \cM^G$ and the gauged theory $(\Sigma\cM)_{\ve[2]}^\phi=\Sigma (\cM_G^\xt)^G$. Here $()^G$ denotes the equivariantization. In the language of Refs.~\cite{LKW1602.05936,LKW1602.05946}, $\cM^G$ is a MTC over $\Rep G$; $(\cM_G^\xt)^G$ is a minimal modular extension of $\cM^G$, which is obtained via gauging the $G$ symmetry.
  \begin{figure}[ht]
    \centering
  \begin{tikzpicture}[baseline=(current bounding box.center)]
      \fill[green!30!white] (0,0)--(3,0) arc [start angle=-90, end angle=-180, radius=1] arc [start angle=0, end angle=90, radius=1]
      arc [start angle=-90, end angle=-180, radius=1] -- cycle;
      \draw[thick,blue] (3,0) arc [start angle=-90, end angle=-180, radius=1] node[above right]{$\Sigma\cM^G$} arc  [start angle=0, end angle=90, radius=1] 
      arc [start angle=-90, end angle=-180, radius=1];
      \draw[thick,blue] (0,0)--node[left]{$\Rep[2] G$} (0,3) node[below left]{$ $} --node[right]{$\Sigma (\cM_G^\xt)^G$} (0,4);
      \draw[thick] (4,0)--node[below]{$\Sigma\cM$} (3,0)node[above]{$\phi$} --node[below]{$\ve[2]_G$} (0,0)node[below left]{$\ve[2]$};
      \node[blue] at (3,3) {$ $};
      \node[blue] at (1,1) {$Z(\ve[2]_G)$};
    \end{tikzpicture}
    \caption{2+1D Bosonic SET}
    \label{fig.3dbset}
  \end{figure}
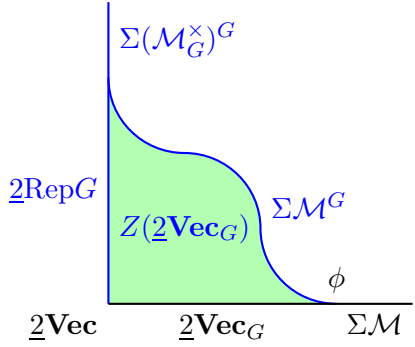
\end{example}
\begin{remark}
     We give a reconstruction from $\phi:\ve[2]_G\to \Sigma\cM$, to the braided $G$-crossed extension $\cM_G^\xt$ via the internal hom
    $$\Hom_{\ve[2]_G}(g,[X,Y])\cong \Hom_{\Sigma\cM}(\phi(g)\ot X,Y).$$
    Let $\one$ be the tensor unit of $\Sigma\cM$. We see $$[\one,\one]\cong \bigoplus\limits_{g\in G}\Hom_{\Sigma\cM}(\phi(g),\one) \in \ve[2]_G$$ is a $G$-graded fusion category with neutral component $\Hom_{\Sigma\cM}(\one,\one)=\cM$. The $G$-crossed braiding is induced from the interchanger together with the invertibility of $\phi(g)$. For $a\in \Hom_{\Sigma\cM}(\phi(g),\one), b\in \Hom_{\Sigma\cM}(\phi(h),\one)$,  we denote $a':=\id_{\phi(g^{-1})} \ot a\in \Hom_{\Sigma\cM}(\one, \phi(g^{-1}))$, 
    thus
    \begin{align*}
    b\ot a' &\cong (\id_\one\circ b)\ot (a'\circ \id_\one)\cong (\id_\one\ot a')\circ (b\ot \id_\one)
    \\
    &\cong (a'\ot\id_\one)\circ(\id_\one\ot b)\cong (a'\circ \id_\one)\ot (\id_\one \circ b)
    \\&\cong a'\ot b.
    \end{align*}
    Further tensoring $\id_{\phi(g)}$ to the left of both sides, we get
    $$\left(\id_{\phi(g)}\ot b\ot \id_{\phi(g^{-1})}\right) \ot a\cong a\ot b.$$
    We can also see that $[\one,\one]$ admits a natural $G$-action: $\id_{\phi(g)}\ot -\ot \id_{\phi(g^{-1})}$ and thus also a natural half-braiding. $[\one,\one]$ therefore lifts to an algebra in $Z(\ve[2]_G)$.    
\end{remark}
\begin{example}
    [Ordinary gauging of 1-symmetry in 2+1D] If the 1-symmetry $\Sigma\cB$ is local, it admits a monoidal functor $\Sigma\cB\to \ve[2]$, and thus a braided functor $\cB\to \ve$. Therefore, $\cB$ must be a Tannakian category, equivalent to $\Rep G$ for some group $G$. It is then clear that ordinary gauging of $\Sigma\cB$ is just the ungauging of Example~\ref{eg.2set}. 
\label{2+1D1-sym}
\end{example}

\begin{example}
    [Partial gauging\TL{: gauging subsysmmetry}]
    \label{eg.subsym}
    Pick a group $G$ and subgroup $L\subset G$, and choose $\cA=\ve[n]$, $\cT=\ve[n]_G$, $\phi:\ve[n]_G\to \ve[n]$ the forgetful functor, $\eta:\ve[n]_L\hookrightarrow \ve[n]_G$, and $\cK=\ve[n]\btr[{\ve[n]_L}]{}_\eta \ve[n]_G\cong \lmd{\ve[n-1]_L}{\ve[n]_G}$. Under the forgetful functor $\phi$, we have identification $_\phi \ve[n]\cong\rmd{\ve[n-1]_G}{ \ve[n]_G }$, thus $\cK\btr[{\ve[n]_G}] {}_\phi \ve[n]\cong \bmd{\ve[n-1]_L}{\ve[n-1]_G}{ \ve[n]_G }\cong \Rep[n] L$.
    The gauged theory is then $(\ve[n])_\cK^\phi\cong Z_{\ve[n]|\ve[n]}(\cK\btr[{\ve[n]_G}] {}_\phi \ve[n])\cong Z_{\ve[n]|\ve[n]}(\Rep[n] L)$, the $L$-gauge theory, equipped with the gauge symmetry $(\ve[n]_G)_\cK^\vee=Z_{\ve[n]|\ve[n]_G}(\ve[n]\btr[{\ve[n]_L}]{}_\eta \ve[n]_G)\cong \Rep[n] L\btr[\ov{Z(\ve[n]_L)}] Z(\eta)=(\ve[n]_G)_{\ve[n]}^\eta$. 
    Alternatively, one can compute the gauged theory using the fact that $\phi\circ\eta:\ve[n]_L\to\ve[n]_G\to\ve[n]$ is the forgetful functor on $\ve[n]_L$.
    
    Moreover, $\Rep[n] L$ embeds into the gauge symmetry $(\ve[n]_G)_\cK^\vee$, which can be seen from two perspectives: 
    \begin{align*}
    \Rep[n] L &\cong \bmd{\ve[n-1]_L}{\ve[n-1]_L}{ \ve[n]_L }
    \\&
    \hookrightarrow\bmd{\ve[n-1]_L}{\ve[n-1]_L}{ \ve[n]_G }\cong (\ve[n]_G)_\cK^\vee,
    \end{align*}
    or 
    \begin{align*}
    \Rep[n] L&\cong \Rep[n] L\btr[\ov{Z(\ve[n]_L)}] Z(\ve[n]_L)
    \\&
    \hookrightarrow \Rep[n] L\btr[\ov{Z(\ve[n]_L)}] Z(\eta: \ve[n]_L\hookrightarrow \ve[n]_G)
    \\&\cong(\ve[n]_G)_\cK^\vee.
    \end{align*}
     See Figure~\ref{fig.pg} for a graphical presentation. In the special case that $G=L\times G/L$, we further have $(\ve[n]_G)_\cK^\vee\cong \Rep[n] L\bt \ve[n]_{G/L}$.

    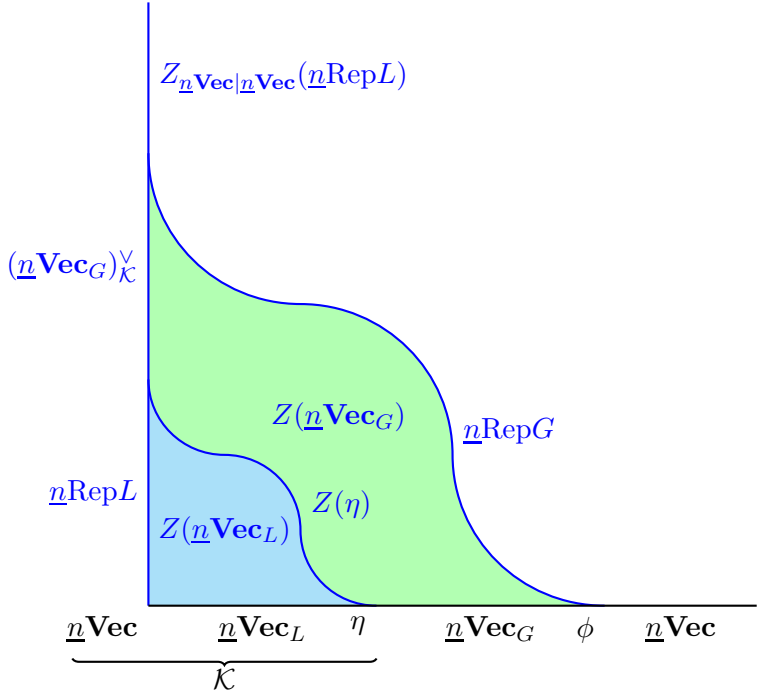
\begin{figure}[ht]
    \centering
    \begin{tikzpicture}[baseline=(current bounding box.center)]
    \fill[green!30!white] (0,0)--(6,0) arc [start angle=-90, end angle=-180, radius=2] arc [start angle=0, end angle=90, radius=2]
    arc [start angle=-90, end angle=-180, radius=2] -- cycle;
    \draw[thick,blue] (6,0) arc [start angle=-90, end angle=-180, radius=2] node[above right]{$\Rep[n] G$} arc  [start angle=0, end angle=90, radius=2] 
    arc [start angle=-90, end angle=-180, radius=2];
    \fill[cyan!30!white] (0,0)--(3,0) arc [start angle=-90, end angle=-180, radius=1] arc [start angle=0, end angle=90, radius=1]
    arc [start angle=-90, end angle=-180, radius=1] -- cycle;
    \draw[thick,blue] (3,0) arc [start angle=-90, end angle=-180, radius=1] node[above right]{$Z(\eta)$} arc  [start angle=0, end angle=90, radius=1] 
    arc [start angle=-90, end angle=-180, radius=1];
    \draw[thick,blue] (0,0) -- node[left]{$\Rep[n] L$} (0,3) -- node[left]{$(\ve[n]_G)_\cK^\vee$} (0,6) node[below left]{$ $} --node[right]{$Z_{\ve[n]|\ve[n]}(\Rep[n] L)$} (0,8);
    \draw[thick] (8,0)--node[below]{$\ve[n]$} (6,0)
    node[below left]{$\phi$}
    --node[below]{$\ve[n]_G$} (3,0) node[below left]{$\eta$}
    --node[below]{$\ve[n]_L$} (0,0)node[below left]{$\ve[n]$};
    \node[blue] at (6,6) {$ $};
    \node[blue] at (2.5,2.5) {$Z(\ve[n]_G)$};
    \node[blue] at (1,1) {$Z(\ve[n]_L)$};
    \draw [ thick, decoration={
        brace,
        mirror,
    },
    decorate
] (-1,-0.7) --node [below] {$\cK$} (3,-0.7)  ; 
  \end{tikzpicture}
    \caption{Gauging a subgroup symmetry as partial gauging}
    \label{fig.pg}
   \end{figure}

\end{example}

\begin{example}
    [\TL{Partial gauging which are not gauging subsymmetries}]\label{eg.pgnonsub}
    \TL{Consider fusion 1-category symmetry $\Rep G$ for some finite group $G$.
    For any subgroup $H\hookrightarrow G$, we have a pullback monoidal functor $\Rep G\to \Rep H$, and $\Rep H$ can be viewed as a right $\Rep G$ module, or an abstract gauging. It corresponds to the algebra $\Fun(G/H)$ of linear functions on the cosets $G/H$. 

    Now consider the fusion category symmetry $\Rep S_3$. The only proper fusion subcategory of $\Rep S_3$ is $\Rep \Z_2$. There is a unique way to completely gauge $\Rep \Z_2$, corresponding to the algebra $\Fun(\Z_2)$. Denote the embedding by $\eta:\Rep \Z_2\hookrightarrow \Rep S_3$, we know $\dim \eta(\Fun(\Z_2))=\dim \Fun(\Z_2)=2$. On the other hand, consider the abstract gauging $\Rep \Z_2$ of $\Rep S_3$ ($\Rep \Z_2$ as a right $\Rep S_3$-module instead of as a subcategory). The right $\Rep S_3$-module $\Rep \Z_2$ corresponds to the algebra $\Fun(S_3/\Z_2)$ in $\Rep S_3$, and $\dim \Fun(S_3/\Z_2)=3$. Since the algebras $\eta(\Fun(\Z_2))$ and $\Fun(S_3/\Z_2)$ are both commutative and $\dim \eta(\Fun(\Z_2))\neq \dim \Fun(S_3/\Z_2)$, they cannot be Morita equivalent. Therefore, the abstract gauging $\Rep \Z_2$ of $\Rep S_3$ is not the complete (neither ordinary nor twisted) gauging of any subsymmetry of $\Rep S_3$.}
\end{example}

\begin{example} [Higher gauging of 1-symmetry in 2+1D topological order]
\label{eg.HigherGauging}
 The full symmetry of a 2+1D topological order can be thought as $\Sigma\cM$ for a MTC $\cM$. We now consider gauging $\cM$ in 1+1D, also known as 1-gauging of 1-symmetry in 2+1D~\cite{RSS2204.02407}. Now forget the braiding, think $\cM$ as a fusion 1-category and take the symmetry assignment as $\phi=\id_\cM:\cM\to \cM$. We know that the gauged theories $\cM_\cK^{\id_\cM}$ for all choices of $\cK$, realize all possible gapped boundaries of $Z(\cM)=\cM\bt\ov{\cM}$, or equivalently, the gapped defects in $\cM$ or objects in $\Sigma\cM$.
\end{example}
\section{Conclusion and outlook}
In this paper, we applied the representation principle to study the category of SET orders. Based on a simple formula
\[ \Fun(\Sigma\cT,\cX)\]
we showed in detail how various properties of SET orders can be derived, with the help of the mathematics of higher linear algebra.

We want to emphasize that the spirit of the representation principle is to find a sufficiently complete target category to represent the symmetry. All the relevant physical observables should be included. In this paper, we consider only symmetries which are described by fusion $n$-categories (over complex numbers). Finite onsite unitary symmetries are included in our framework, but still many important physical symmetries are beyond our scope:
\begin{enumerate}
    \item  For fusion $n$-category symmetries, it seems the higher vector spaces, which are Karoubi or condensation complete higher categories including all topological observables, are sufficiently complete, as has been demonstrated in this paper. 
    \item Antiunitary symmetries are essentially $\R$-linear instead of $\C$-linear. In other words, one needs to consider at least the higher vector spaces over real numbers. Unlike $\C$, $\R$ is not algebraically closed, which makes the representation theory technically much more complicated.
    \item For lattice symmetries, probably one has to taken into account, in a more precise manner, how the symmetry acts on a microscopic lattice, as emphasized in~\cite{YZ2309.15118}. In other words, the data describing the microscopic lattice are also physical observables that should be included and for a complete representation category, one need to consider how the symmetry is represented on the lattice.
    \item Even more generally for spacetime symmetries and continuous symmetries, the complete target category should at least include the metric of spacetime or other continuous structures and the representation should also be continuous in appropriate sense. 
\end{enumerate}
\medskip

We believe that in the cases not covered in this paper, the representation principle would still work, and if not, it should be reflected whether all relevant physical observables have been taken into account to form a complete mathematical structure, in which the symmetry is properly represented. 
\medskip

\begin{acknowledgments}
We are grateful to Holiverse Yang for discussions. TL is supported by start-up funding from The Chinese University of Hong Kong, and by funding from Research Grants Council, University Grants Committee of Hong Kong (ECS No.~24304722).
\end{acknowledgments}

\newpage
\appendix
\section{Karoubi Completion}
In this section, we briefly review the Karoubi Completion theory of an $n$-category, developed by Davide Gaiotto and Theo Johnson-Freyd in \cite{GJ1905.09566}.
\begin{definition}[Condensation]
    Let $\cC$ be an $n$-category, $X,Y\in \cC$. For $n=0$, a condensation $X\cto Y$ means $X=Y$. Inductively, for $n>1$, a condensation $X\cto Y$ is a pair of 1-morphism $f: X\rightarrow Y$, $g: Y\rightarrow X$ such that $f\circ g\cto \id_Y$. If $X\cto Y$, we say $Y$ is a condensate of $X$.
    \end{definition}
    \begin{remark}
    Physically, consider two $n$D topological phases $X,Y$. $X\cto Y$ mimics a phase transition from $X$ to $Y$. Some $Y$ regions arise from $X$ with domain walls $f_1$ and $g_1$. When $X$ transits to $Y$, these $Y$ regions begin to connect to each other and become a larger $Y$-region, which means the domain wall $g_1$ and $f_1$ fuse together and leaves  1-codimensional higher defects $f_2,g_2$ between $f_1\circ g_1$ and the trivial codimeion-1 defect in $Y$, i.e. $f_1\circ g_1\cto \id_Y$. 
    After similar processes for higher codimensions, finally $X$ is completely replaced by $Y$. See Figure~\ref{fig.condensation}.
    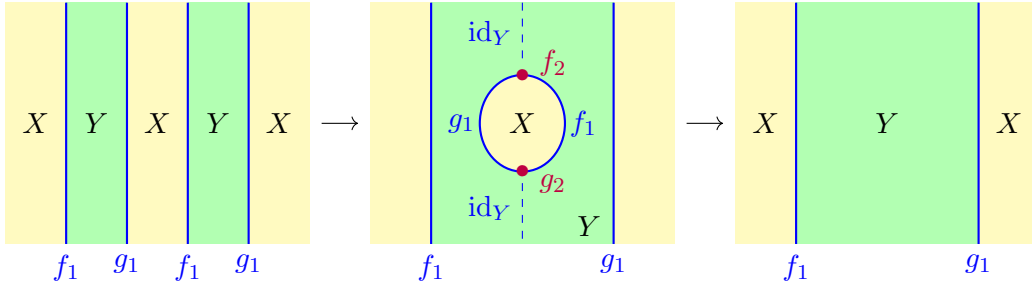
\begin{figure}[ht]
  \centering
\begin{tikzpicture}[scale=0.8]
  \fill[yellow!30!white] (0,0) rectangle (1,4);
   \node[black] at (0.5,2) {$X$};
  \fill[green!30!white] (1,0) rectangle (2,4);
  \node[black] at (1.5,2){$Y$};
  \fill[yellow!30!white] (2,0) rectangle (3,4);
  \fill[green!30!](3,0) rectangle (4,4);
  \fill[yellow!30!white](4,0) rectangle (5,4);
    \node[black] at (2.5,2){$X$};
  \node[black] at (3.5,2){$Y$};
  \node[black] at (4.5,2){$X$};
   \draw[thick,blue](1,0)node[below]{$f_1$}--(1,4);
  \draw[thick,blue](2,0)node[below]{$g_1$}--(2,4);
   \draw[thick,blue](3,0)node[below]{$f_1$}--(3,4);
    \draw[thick,blue](4,0)node[below]{$g_1$}--(4,4);
  \draw[->] (5.2,2)--(5.8,2);
  \fill[yellow!30!white] (6,0) rectangle (7,4);
  \fill[green!30!white](7,0) rectangle (10,4);
  \draw[thick,blue] (7,0)node[below]{$f_1$}--(7,4);
  \node[blue] at (7.5,2){$g_1$};
  \node[blue] at (9.5,2){$f_1$};
  \draw[dashed,blue] (8.5,4)--(8.5,3.5)node[left]{$\id_Y$}--(8.5,3);
  \draw[dashed,blue] (8.5,1.8)--(8.5,0.6)node[left]{$\id_Y$}--(8.5,0);
  \fill[yellow!30!white] (8.5,2) ellipse (0.7 and 0.8);
   \draw[thick,blue](8.5,2) ellipse (0.7 and 0.8);
  \node[purple] at (8.5,2.8){$\bullet$};
  \node[purple] at (9,3){$f_2$};
  \node[purple] at (8.5,1.2){$\bullet$};
  \node[purple] at (9,1){$g_2$};
  \node[black] at  (8.5,2){$X$};
  \node[black] at (9.6,0.3){$Y$};
  \fill[yellow!30!white](10,0) rectangle (11,4);
  \draw[thick,blue] (10,0)node[below]{$g_1$}--(10,4);
  \draw[->](11.2,2)--(11.8,2);
  \fill[yellow!30!white](12,0) rectangle (17,4);
  \fill[green!30!white](13,0) rectangle (16,4);
  \draw[thick,blue] (13,0)node[below]{$f_1$}--(13,4);
  \draw[thick,blue] (16,0)node[below]{$g_1$}--(16,4);
  \node[black] at (12.5,2){$X$};
  \node[black] at (14.5,2){$Y$};
  \node[black] at (16.5,2){$X$};
\end{tikzpicture}
\caption{A phase transition induced by condensation.}
\label{fig.condensation}
\end{figure}
 \end{remark}
\begin{definition}(Walking $n$-condensation)
  The walking $n$-condensation is an $n$-category $\spadesuit_n$ freely generated by an $n$-condesation, i.e. there are two objects $X,Y \in \spadesuit_n$, 1-morphisms are generated by $Y\xrightarrow{g_1}X\xrightarrow{f_1} Y$, 2-morphisms are generated by $\id_Y\xrightarrow{g_2}f_1\circ g_1\xrightarrow{f_2}\id_Y$, so on and so forth, and $n$-morphisms are generated by 
  $\id_{\cdots_{\id_Y}}\xrightarrow{g_n} f_{n-1}\circ g_{n-1}\xrightarrow{f_n}\id_{\cdots_{\id_Y}} $ satisfying 
  \begin{equation}
      f_n\circ g_n = \id_{\cdots_{\id_Y}}. 
  \end{equation}
    Let $\cC$ be an $n$-category. An $n$-condensation in $\cC$ is a functor  $F:\spadesuit_n\rightarrow \cC$.
\end{definition}

\begin{definition}(Walking $n$-condensation monad) The walking $n$-condensation monad $\clubsuit_n$ is the full subcategory of $\spadesuit_n$ restricted on the object $X$. Let $\cC$ be an $n$-category. An $n$-condensation monad in $\cC$ is a functor  $G:\clubsuit_n\rightarrow \cC$.
\end{definition}

\begin{definition}(Karoubi complete/condensation complete)
A $0$-category is always assumed to be Karoubi complete. Inductively, for $n\geq 1$, an $n$-category $\cC$ is Karoubi complete or condensation complete if $\forall X,Y\in \cC$, $\Hom_\cC(X,Y)$ is Karoubi complete, and if any $n$-condensation monad in $\cC$ extends to an $n$-condensation in $\cC$.
\end{definition}

Given an $n$-category $\cC$ whose hom categories are Karoubi complete, we can perform the so-called Karoubi completion of $\cC$, resulting in a Karoubi complete $n$-category $\mathrm{Kar}(\cC)$, whose objects are all the condensates of objects in $\cC$. To be more precise, objects in $\mathrm{Kar}(\cC)$ are condensation monads in $\cC$, 1-morphisms are condensation bimodules, and higher morphisms are higher bimodule maps. We refer readers to \cite{GJ1905.09566} for detailed constructions.

\begin{example}
For a 1-category $\cC$, $\mathrm{Kar}(\cC)$ is the idempotent completion of $\cC$.     
\end{example}

\begin{example}
    For a finite semisimple monoidal 1-category $\cC$, $\Sigma\cC=\mathrm{Kar}(B\cC)$ has a following explicit description, objects are separable algebras in $\cC$, $1$-morphisms are bimodules over algebras, and $2$-morphisms are bimodule maps. One can see~\cite{xi2024,YWL2024} for   examples on $\cC=\sve, \ve_{\mathbb Z_2 \times \mathbb Z_2 }, \ve_{\mathbb Z_4}$.
\end{example}

\section{(Relative) tensor product}\label{sec.tensor}
Let $\mathrm{KarCat_n}$ denote the $(n+1)$-category of Karoubi-complete $\C$-linear $n$-categories. 
In Section~\ref{sec.hl}, we define the  tensor product of $\cC,\cD\in \mathrm{KarCat_n}$ by the object that representing the functor $\Hom_{\mathrm{KarCat_n}}(\cC,\Hom_{\mathrm{KarCat_n}}(\cD,-))$. We sketch a more practical inductive definition, which was given in \cite{Joh2003.06663}.

\begin{definition}[Tensor product in $\mathrm{KarCat_n}$]
Let $\boxtimes$ in $\mathrm{KarCat_0}$ be the tensor product of vector spaces. Suppose $\boxtimes$ in $\mathrm{KarCat_{n-1}}$ is well defined, for $\cC,\cD\in \mathrm{KarCat_n}$ we can then define the naive tensor product $\cC\otimes \cD$ be the $n$-category as follows
\begin{enumerate}
    \item object set:
    \begin{equation*}
        \mathrm{obj}(\cC\otimes \cD) = \mathrm{obj}(\cC) \times \mathrm{obj}(\cD),
    \end{equation*}
 \item home spaces: $\forall x_1,x_2\in \cC$, $y_1,y_2 \in \cD$
 \begin{equation*}
\hom_{\cC\otimes \cD}((x_1,y_1),(x_2,y_2)) = \hom_{\cC}(x_1,x_2)\boxtimes \hom_\cD(y_1,y_2).
 \end{equation*}
\end{enumerate}
Note that $\hom_{\cC}(x_1,x_2), \hom_\cD(y_1,y_2) \in \mathrm{KarCat_{n-1}}$, their Deligne tensor product is supposed to be well defined.
Then 
\begin{equation*}
    \cC\boxtimes \cD=\mathrm{Kar}(\cC\otimes \cD).
\end{equation*}
\end{definition}
\begin{remark}
    Objects in $\cC\boxtimes \cD$  are arbitrary direct sums and condensates generated from $x\boxtimes y, \forall x\in \cC, y\in \cD$. These direct sums and condensates are in general not in $\cC\otimes \cD$.
\end{remark}

\begin{definition}[balanced functor]
Let $\cC$ be a fusion 1-category, and $\cM, \cN$ be a right $\cC$-module category and a left $\cC$-module category with module action and module associator $(\triangleleft: \cM\times \cC \rightarrow \cM, \beta^\cM), (\triangleright: \cC\times \cN\rightarrow \cN, \beta^\cN)$, respectively. Let $\cA$ be a linear finite semisimple category. A balanced functor $F: \cM\times \cN\rightarrow \cA$ is a bilinear functor together with a natural isomorphism 
\begin{equation*}
    \alpha: F(\triangleleft\times \id_\cN)\cong F(\id_\cM\times \triangleright)
\end{equation*}
    satisfying the pentagon equation
    \begin{equation*}
        \begin{tikzcd}
	{F(m\triangleleft(x\otimes y),n)} && {F(m,(x\otimes y)\triangleright n)} \\
	{F((m\triangleleft x)\triangleleft y,n)} && {F(m,x\triangleright(y\triangleright n))} \\
	& {F(m\triangleleft x, y\triangleright n)}
	\arrow["{\alpha_{m,(x\otimes y),n}}", from=1-1, to=1-3]
	\arrow["{F(\beta^{\cM}_{m,x,y}, \id_n)}"', from=1-1, to=2-1]
	\arrow["{F(\id_m,\beta^\cN_{x,y,n})}", from=1-3, to=2-3]
	\arrow["{\alpha_{m\triangleleft x,y,n}}"', from=2-1, to=3-2]
	\arrow["{\alpha_{m,x,y\triangleright n}^{-1}}", from=2-3, to=3-2]
\end{tikzcd}
    \end{equation*}
and triangle equation 
\begin{equation*}
    \begin{tikzcd}
	{F(m\triangleleft \one_\cC,n)} && {F(m,\one_\cC\triangleright n)} \\
	& {F(m,n)}
	\arrow["{\alpha_{m,\one_\cC,n}}", from=1-1, to=1-3]
	\arrow["{F(\mathrm{unit}_m,\id_n)}"', from=1-1, to=2-2]
	\arrow["{F(\id_m, \mathrm{unit}_n)}", from=1-3, to=2-2].
\end{tikzcd}
\end{equation*}
$\forall x,y\in \cC, m\in \cM, n\in \cN$.
\end{definition}

\begin{definition}[Relative tensor product of module 1-categories]\label{RTPn=1}
Let $\cC$ be a multi-fusion 1-category and $\cM,\cN$ be a right $\cC$ module category and left $\cC$ module category, respectively. We define the relative tensor product
$\cM\btr[\cC]\cN$ as the linear finite semisimple category together with a balanced functor $\btr[\cC]$, such that for any balanced functor $F: \cM\times \cN\rightarrow \cA$, there exists a unique linear functor $\Tilde{F}$ such that
\begin{equation*}
\Tilde{F}\circ\btr[\cC] \cong F
\end{equation*}
as balanced functors, i.e. the following diagram 
\begin{equation*}
\begin{tikzcd}
	{\cM\times \cN} \\
	{\cM\btr[\cC]\cN} & \cA
	\arrow["{\btr[\cC]}"', from=1-1, to=2-1]
	\arrow["{\forall F}", from=1-1, to=2-2]
	\arrow["\exists  !{\Tilde{F}}"', from=2-1, to=2-2]
\end{tikzcd}
\end{equation*}
commutes up to a balanced functor natural isomorphism.
\end{definition}

\begin{remark}
Let $A,B$ be algebras in $\cC$ such that $\cM\cong \lmd A{\cC}$ and $\cN\cong \rmd B{\cC}$ as right and left $\cC$-module categories, respectively.
    In \cite{Douglas_2019}, the authors proved that the category of 
$A$-$B$ bimodules $\bmd A B{\cC}$ realizes the relative tensor product $\cM\btr[\cC]\cN$, i.e, 
\begin{equation*}
    \cM\btr[\cC]\cN \cong \lmd A{\cC}\btr[\cC]\rmd B{\cC} \cong \bmd A B{\cC}.
\end{equation*}
For higher category cases in this paper, we also use higher category of bimodules to characterize the relative tensor product.
\end{remark}

\begin{definition}[Relative tensor product of module $n$-categories]\label{RTPn}
 Let $\cC$ be multi-fusion $n$-category, $\cM$ and $\cN$ be a right $\cC$-module category and a left $\cC$-module category, respectively. The relative tensor product $\cM\btr[\cC]\cN$ is defined to be the colimit of the following diagram
 \begin{equation*}
     \begin{tikzcd}
	{...} & {\cM\times\cC\times\cC\times\cN} & {\cM\times\cC\times\cN} & {\cM\times\cN}
	\arrow[shift left=3, from=1-1, to=1-2]
	\arrow[shift right=3, from=1-1, to=1-2]
	\arrow[shift left, from=1-1, to=1-2]
	\arrow[shift right, from=1-1, to=1-2]
	\arrow[shift left=2, from=1-2, to=1-3]
	\arrow[shift right=2, from=1-2, to=1-3]
	\arrow[from=1-2, to=1-3]
	\arrow[shift left, from=1-3, to=1-4]
	\arrow[shift right, from=1-3, to=1-4]
\end{tikzcd}.
 \end{equation*}
\end{definition}

\begin{remark}
 Each arrow on the above diagram means a functor from left to right. For example, the three arrows from $\cM \times \cC\times \cC\times \cN$ to $\cM\times \cC \times \cN$ means 
$(\triangleleft,\id_{\cC},\id_{\cN})$,$(\id_{\cM},\ot,\id_{\cN})$, and $(\id_{\cM},\id_{\cC},\triangleright)$. There are also (higher) cells between these arrows given by (higher) coherence morphisms such as the associator. If we add the arrow
\begin{equation*}
 \cM\times \cN\xrightarrow{\btr[\cC]} \cM\btr[\cC]\cN 
\end{equation*}
in the end, then the composition of $\btr[\cC]$ with different paths of arrows will result in equivalent functors. Also, although this is in general an infinite long diagram, the first $n+2$ terms are sufficient  to define the relative tensor product of two module $n$-categories. For $n=1$, the colimit of the diagram
\begin{equation*}
     \begin{tikzcd}
{\cM\times\cC\times\cC\times\cN} & {\cM\times\cC\times\cN} & {\cM\times\cN}
	\arrow[shift left=2, from=1-1, to=1-2]
	\arrow[shift right=2, from=1-1, to=1-2]
	\arrow[from=1-1, to=1-2]
	\arrow[shift left, from=1-2, to=1-3]
	\arrow[shift right, from=1-2, to=1-3]
\end{tikzcd}
 \end{equation*}
 exactly matches Definition~\ref{RTPn=1} for $\cM\btr[\cC]\cN$.
\end{remark}
Relative tensor product might be understood as the dimension reduction of topological phases. For example, consider $\cC$ as the braided fusion $(n-1)$-category of codimension-2 and higher defects in an $n$+1D topological order. Let $\cM$, $\cN$ be the fusion $n$-categories of codimension-1 and higher defects in two gapped boundaries of $\cC$. The defects in the bulk can fuse with the boundary defects into some new boundary defects, which equips $\cM$ and $\cN$ with right and left $\cC$-module structures, respectively (see Figure~\ref{fig.dr}). Then we do a dimension reduction i.e. squeeze the $\cM,\cC,\cN$ sandwich and get an $n$D topological order, whose codimension-1 and higher defects form the category $\cM\btr[\cC]\cN$.
\begin{figure}[ht]
  \centering
\begin{tikzpicture}
  \fill[yellow!30!white] (0,0) rectangle (4,2);
  \draw[thick,blue] (0,0)--(0,1)node[left]{$\cM$} --(0,2) ;
  \draw[thick,blue] (4,0)-- (4,1)node[right]{$\cN$} --(4,2);
  \node[blue] at (2,1) {$\cC$};
  \draw[->] (4.6, 1)--(5.7,1)node[above]{reduction}--(6.8, 1);
  \draw[thick,blue] (7,0)--(7,1)node[right]{$\cM\btr[\cC]\cN$}--(7,2);
\end{tikzpicture}
\caption{Relative tensor product and dimension reduction.}
\label{fig.dr}
\end{figure}
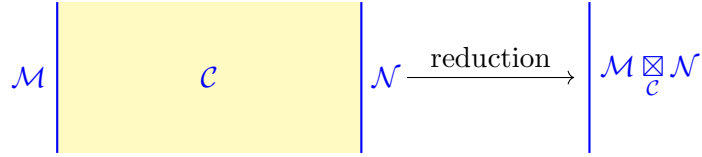

\section{Relative center}\label{sec.rc}
\begin{definition}
    [Relative Center] Let $\cT$, $\cA$ be fusion categories, and $\phi:\cT\rightarrow \cA$ is a tensor functor, then the relative center of $\phi$ is a fusion category $Z_{\cT}(\cA)$. The objects in $Z_{\cT}(\cA)$ are pairs $(X,\gamma)$, where $X\in \cA$, and $\gamma:  \phi(-)\otimes X\cong X\otimes \phi(-)$  is a collection of half braiding natural  isomorphisms, such that the diagram
    \begin{equation*}
         \begin{tikzcd}
	{\phi(s\otimes t )\otimes X} && {X\otimes \phi(s\otimes t )} \\
	{\phi(s)\otimes \phi(t)\otimes X} && {X\otimes\phi(s)\otimes\phi(t)} \\
	& {\phi(s)\otimes X\otimes \phi(t)}
	\arrow["{\gamma_{s\otimes t}}", from=1-1, to=1-3]
	\arrow["{\phi_{s,t}}"', from=1-1, to=2-1]
	\arrow["{\gamma_t}"', from=2-1, to=3-2]
	\arrow["{\phi_{s,t}}", from=1-3, to=2-3]
	\arrow["{\gamma_s}"', from=3-2, to=2-3]
\end{tikzcd}
    \end{equation*}
      
    commutes $\forall s,t\in \cT$. Here $\phi_{s,t}:\phi(s\otimes t)\cong \phi(s)\otimes \phi(t)$ is the tensor functor structure of $\phi$ and associators are omitted. 

    The monoidal structure of $Z_\cT(\cA)$ is given by 
    \begin{equation*}
         (X,\gamma)\otimes (X',\gamma') = (X\otimes X', \gamma\otimes\gamma')
    \end{equation*} 
    where $(\gamma\otimes\gamma')_s$ is given by the composition \begin{equation*}
        \phi(s)\otimes X\otimes X'  \xrightarrow[ ]{\gamma_s} X\otimes \phi(s)\otimes X' \xrightarrow{\gamma'_s}X\otimes X'\otimes\phi(s) 
    \end{equation*}
\end{definition}
\begin{definition}
    (Relative Center, an alternative definition) Let $\cT$ be a fusion category and $\cA$ is a $\cT$-$\cT$ bimodule, then the relative center of $\cA$ is defined as $\Fun_{\cT|\cT}(\cT,\cA)$.
\end{definition}
\begin{proposition} Let $\cT$, $\cA$ be fusion categories, and $\phi:\cT\rightarrow \cA$ is a tensor functor, then 
    $$Z(\phi)\cong \Fun_{\cT|\cT}(\cT,{}_\phi\cA_\phi)\cong Z_{\cT}(\cA).$$
\end{proposition}
\begin{proof}
Here we sketch the proof. By definition, $Z(\phi) =Z_{\cT|\cA}({}_\phi\cA)= \Fun_{\cT|\cA}({}_\phi\cA,{}_\phi\cA)$, then the first equivalence and its quasi-inverse are   
\begin{equation*}
\begin{split}
\text{Fun}_{\cT|\cA}(_\phi\cA, {}_\phi\cA)&\cong\text{Fun}_{\cT|\cT}(\cT,{}_\phi\cA_\phi)  \\
F &\mapsto \phi(-)\otimes  F(\one_\cA)\cong F(\phi(-))\\
G(\one_{\cT})\otimes - &\mapsfrom G
\end{split}
\end{equation*}
the module functor structure of $F(\phi(-))$ and $G(\one_\cT)\otimes -$ are inherited from the module functor structure of $F$ and $G$. The invertibility is then
\begin{equation*}
\begin{split}
     &F(\phi(\one_\cT))\otimes - \cong F(\one_\cA\otimes - ) \cong F;\\
     &G(\one_\cT)\otimes \phi(-) \cong G(\one_\cT\otimes -) \cong G
\end{split}
\end{equation*}

The second equivalence and its quasi-inverse are 
\begin{equation*}
\begin{split}
\text{Fun}_{\cT|\cT}(\cT, {}_\phi\cA_\phi)& \cong Z_{\cT}(\cA)\\
    G&\mapsto G(\one_{\cT})\\
    \phi(-)\otimes X &\mapsfrom X
\end{split}
\end{equation*}
The half braiding of $G(\one_{\cT})$ is given by
\begin{equation*}
    \phi(s)\otimes G(\one_{\cT})\simeq G(s\otimes \one_{\cT})\simeq G(\one_{\cT}\otimes s)\simeq G(\one_{\cT})\otimes \phi(s).
\end{equation*}
The left $\cT$ module functor structure of $\phi(-)\otimes X$ is given by 
\begin{align*}
 \phi(s_1\otimes s_2)\otimes X \simeq \phi(s_1)\otimes (\phi(s_2)\otimes X)
\end{align*}
since $\phi$ is a monoidal functor. It's right $\cT$ module functor structure is given by 
\begin{align*}
    &\phi(s_1\otimes s_2)\otimes X \simeq \phi(s_1)\otimes(\phi(s_2)\otimes X)
    \\&\simeq \phi(s_1)\otimes (X\otimes \phi(s_2))\simeq (\phi(s_1)\otimes X)\otimes \phi(s_2),
\end{align*}
the second isomorphism is the half braiding of $X$ with $\phi(s)$. The invertibility is 
\begin{equation*}
    \begin{split}
&\phi(-)\otimes G(\one_\cT)\simeq G(-\otimes \one_\cT)\simeq G; \\
&\phi(\one_\cT)\otimes X\simeq \one_\cA \otimes X\simeq X        
    \end{split}
\end{equation*}

\end{proof}





\bibliographystyle{JHEP}
\bibliography{biblio}

\end{document}